\documentclass[reqno,11pt]{amsart}
\usepackage[mathscr]{eucal}
\usepackage{slashed}
\usepackage{bbm}

\numberwithin{equation}{section}

\title[The Rescaled Causal Perturbation Expansion]
{The Causal Perturbation Expansion Revisited: \\
Rescaling the Interacting Dirac Sea}

\author[F.\ Finster]{Felix Finster}
\author[A.\ Grotz]{Andreas Grotz \\ \\ January 2009}

\address{NWF I - Mathematik \\ Universit\"at Regensburg \\ D-93040 Regensburg \\ Germany}
\email{Felix.Finster@mathematik.uni-regensburg.de}
\email{Andreas.Grotz@mathematik.uni-regensburg.de}

\thanks{Supported in part by the Deutsche Forschungsgemeinschaft.}

\newtheorem{Def}{Definition}[section]
\newtheorem{Thm}[Def]{Theorem}
\newtheorem{Prp}[Def]{Proposition}
\newtheorem{Lemma}[Def]{Lemma}
\newtheorem{Remark}[Def]{Remark}
\newtheorem{Corollary}[Def]{Corollary}

\newcommand{\Thanks}{\vspace*{.5em} \noindent \thanks}
\newcommand{\beq}{\begin{equation}}
\newcommand{\eeq}{\end{equation}}

\newcommand{\Proof}{\begin{proof}}
\newcommand{\QED}{\end{proof} \noindent}
\newcommand{\QEDrem}{\ \hfill $\Diamond$}

\newcommand{\bra}{\mbox{$< \!\!$ \nolinebreak}}
\newcommand{\ket}{\mbox{\nolinebreak $>$}}
\newcommand{\C}{\mathbb{C}}
\newcommand{\R}{\mathbb{R}}
\newcommand{\1}{\mbox{\rm 1 \hspace{-1.05 em} 1}}

\newcommand{\N}{\mathbb{N}}

\newcommand{\slsh}{\mbox{ \hspace{-1.1 em} $/$}}

\renewcommand{\O}{{\mathscr{O}}}

\newcommand{\B}{{\mathscr{B}}}

\newcommand{\pvm}{\frac{\pv}{m-m'}}

\newcommand{\sea}{\text{sea}}

\DeclareMathOperator{\pv}{PP}
\DeclareMathOperator{\supp}{supp}
\DeclareMathOperator{\res}{res}
\DeclareMathOperator{\he}{he}
\DeclareMathOperator{\lec}{le}

\DeclareMathOperator{\causal}{causal}
\DeclareMathOperator{\Texp}{Texp}
\allowdisplaybreaks[1]

\begin{document}
\maketitle

\begin{abstract}
The causal perturbation expansion defines the Dirac sea in the presence of
a time-dependent external field. It yields an operator whose image 
generalizes the vacuum solutions of negative energy
and thus gives a canonical splitting of the solution space into two subspaces.
After giving a self-contained introduction to the ideas and techniques, we show that this operator
is in general not idempotent. We modify the standard construction by a rescaling procedure
giving a projector on the generalized negative-energy subspace. The resulting rescaled causal
perturbation expansion uniquely defines the fermionic projector in terms of a series of
distributional solutions of the Dirac equation. The technical core of the paper is
to work out the combinatorics of the expansion in detail.
It is also shown that the fermionic projector with interaction can be obtained from the
free projector by a unitary transformation.
We finally analyze the consequences of the
rescaling procedure on the light-cone expansion.
\end{abstract}

\tableofcontents

\section{Introduction} \label{sec1}
Shortly after the formulation of the Dirac equation~\cite{dirac1}, it was noticed that this
equation has solutions of negative energy, which have no obvious physical interpretation
and lead to conceptual and mathematical difficulties.
Dirac suggested to solve this problem by assuming that in the physical vacuum
all states of negative energy are occupied by electrons forming the so-called 
{\em{Dirac sea}}~\cite{dirac2, dirac3}. Since this many-particle state is homogeneous and
isotropic, it should not be accessible to measurements. Due to the Pauli exclusion principle,
additional particles must occupy states of positive energy, thus being observable as electrons.
Moreover, the concept of the Dirac sea led to the prediction of anti-particles. Namely,
by taking out particles of negative energy, one can generate ``holes'' in the Dirac sea, which are
observable as positrons.

Today, Dirac's intuitive concept of a ``sea of interacting particles'' is often not taken literally.
In {\em{perturbative quantum field theory}}, the problem of the negative-energy solutions is bypassed
by a formal replacement and re-interpretation of the creation and annihilation operators
of the negative-energy states of the free Dirac field, giving rise to a positive definite Dirac
Hamiltonian on the fermionic Fock space.
In the subsequent perturbation expansion in terms of Feynman diagrams, the Dirac sea no longer appears.
This procedure allows to compute the $S$-matrix in a scattering process and gives rise to
the loop corrections, in excellent agreement with the high-precision tests of quantum electrodynamics.

One shortcoming of the perturbative approach is that the particle interpretation of a quantum state
gets lost for intermediate times. This problem becomes apparent already in the presence of
a {\em{time-dependent external field}}. Namely, as first observed by Fierz and Scharf~\cite{fierz+scharf},
the Fock representation must be adapted to the external field as measured by a local 
observer. Thus the Fock representation becomes time and observer dependent,
implying that also the distinction between particles and anti-particles loses its invariant meaning.
The basic problem can be understood already in the one-particle picture:
Mathematically, the distinction between particles and anti-particles corresponds to a splitting of
the solution space of the Dirac equation into two subspaces. In the vacuum, or more generally
in the presence of a static external field $\mathscr{B}(\vec{x})$, in the Dirac equation
\[ (i\gamma^{\mu}\partial_{\mu} + {\mathscr{B}}(\vec{x}) - m)\,\Psi(x)=0 \]
one can separate the time dependence with the plane-wave ansatz
\[ \Psi(t,\vec{x})=e^{-i\omega t}\:\psi(\vec{x})\:. \]
The separation constant~$\omega$, having the interpretation as the energy of the state,
gives a natural splitting of the solution space into solutions of positive and negative energy.
The Dirac sea can be introduced by occupying all states of negative energy.
However, if the external field is time dependent,
\beq \label{tdep}
(i\gamma^{\mu}\partial_{\mu} + {\mathscr{B}}(t, \vec{x}) - m)\, \Psi(x)=0 \:,
\eeq
the separation ansatz no longer works, corresponding to the fact that the energy of the Dirac states
is no longer conserved. Hence the concept of positive and negative energy solutions breaks down,
and the natural splitting of the solution space seems to get lost.

Another shortcoming of the standard reinterpretation of the free Dirac states of negative energy
is that this procedure leads to inconsistencies when the interaction is taken into account on a
{\em{non-perturbative}} level.
For example, in~\cite{hainzl+sere2, hainzl+sere1} the vacuum state is constructed
for a system of Dirac particles with electrostatic interaction in the Bogoliubov-Dirac-Fock approximation.
In simple terms, the analysis shows
that the interaction ``mixes'' the states in such a way that
it becomes impossible to distinguish between the particle states and the states of the Dirac sea.
Thus the only way to obtain a well-defined mathematical setting is to take into account all the
states forming the Dirac sea, with a suitable ultraviolet regularization.

The framework of the fermionic projector is an approach to formulate quantum field theory
in such a way that the above-mentioned problems disappear. Out of all the states of the
Dirac sea we build up the so-called {\em{fermionic projector}}, which puts
Dirac's idea of a ``sea of interacting particles'' on a rigorous mathematical basis.
The fermionic projector gives a global (i.e.\ observer-independent) interpretation of particles and
anti-particles even in the time-dependent setting at intermediate times.
The interaction is described by an action principle, which can be formulated without referring
to the causal or topological structure of the underlying space-time, thus
giving a possible approach for physics on the Planck scale (see~\cite{lrev} for a review
and more references).
When analyzed in Minkowski space, this action principle yields all the Feynman diagrams
of perturbative quantum field theory, but also gives rise to other surprising higher order
corrections~\cite{sector}.
The foundations of this approach have been worked out in the book~\cite{PFP}.
The connection to the fermionic Fock space formalism and to second quantized bosonic fields
is elaborated in~\cite{entangle}.

The construction of the fermionic projector is based on the observation first made in~\cite{sea} that for
the time-dependent Dirac equation~\eqref{tdep} there still is a natural splitting of the solution space
into two subspaces, if one works instead of the sign of the energy
with the underlying {\em{causal structure}}. To explain the basic idea
(for details see Section~\ref{sec2} below),
we first note that the image of the operator~$P^{\sea}_m$ with integral kernel
\beq \label{Pvac}
P^{\sea}_m(x,y) = 
\int \frac{d^4q}{(2\pi)^4}\:(q_\mu \gamma^\mu+m)\: \delta(q^2-m^2)\: \Theta(-q_0)\: e^{-iq(x-y)}
\eeq
coincides precisely with all the negative-energy solutions of the free Dirac equation.
This operator can be decomposed as
\beq \label{Psea}
P^{\sea}_m = \frac{1}{2} \left( p_m - k_m \right) ,
\eeq
where the operator~$k_m$ is causal (in the sense that its kernel is supported inside the light cone),
and the operator~$p_m$ can be understood as the absolute value of the operator~$k_m$. In the case with general interaction, one can extend $k_m$ uniquely to an operator~$\tilde{k}_m$ using the causality property. Making sense of the absolute value, one can also generalize~$p_m$ to an operator~$\tilde{p}_m$.
Introducing in analogy to~\eqref{Psea} the operator
\beq \label{Pint}
P^{\sea}_m = \frac{1}{2} \left( \tilde{p}_m - \tilde{k}_m \right) ,
\eeq
the image of this operator describes the Dirac sea in the case with interaction.

The aim of the present paper is to clarify the normalization of the states of the generalized Dirac
sea~\eqref{Pint}. More precisely, the operator~\eqref{Pvac} describing the free Dirac sea has the
property that it is idempotent if a $\delta$-normalization in the mass parameter is used,
\begin{equation}
P^{\sea}_m P^{\sea}_{m'} = \delta(m-m')\: P^{\sea}_m\:.
	\label{P-proj}
\end{equation}
This idempotence property plays an important role in the framework of the fermionic projector.
However, as we shall see, the interacting Dirac sea~\eqref{Pint} as defined in~\cite{sea}
is in general {\em{not}} idempotent in this sense.
Our goal is to modify the normalization of the states using a rescaling procedure
such as to arrange~\eqref{P-proj}.
This issue of normalizing the states of the interacting Dirac sea can be regarded as a
{\em{problem of functional analysis}}. 
To see the analogy, if the Dirac operator were a self-adjoint operator on a Hilbert space,
we could interpret the product~$P^{\sea}_m \,dm$ as the operator-valued spectral measure of
the Dirac operator, composed by a projector on the generalized negative-energy solutions of the
Dirac equation. Unfortunately, the Dirac operator is only symmetric with respect to the
{\em{indefinite inner product}}
\beq \label{stip}
\bra \Psi | \Phi \ket = \int \overline{\Psi}(x) \Phi(x)\: d^4x \:,
\eeq
making it impossible to use spectral theory in Hilbert spaces.
This is the reason why we must rely on perturbative techniques and work with
formal power series expansions. The main technical task is to work out the combinatorics
of the perturbation expansions in detail. The interesting point is that the details of these
expansions have a correspondence to general results known from functional analysis in Hilbert spaces.
In particular, we recover the polar decomposition, the
resolvent identity, Stone's formula and the functional calculus
from our perturbation expansions (see~\eqref{p-absk}, \eqref{smmp}, Remark~\ref{remark}
and~\eqref{ptil-formally}). We also relate the free and interacting operators by an operator~$U$
which is unitary with respect to the indefinite inner product~\eqref{stip}
(see~\eqref{result-unitary}).

The main result of the paper is the derivation of a unique perturbation expansion for the
fermionic projector which satisfies~\eqref{P-proj}
(see Theorem~\ref{maintheorem}). The summands of this
expansion can be regarded as {\em{Feynman tree diagrams}}, as they also appear in
the standard perturbation expansion in the presence of an external field.
However, our perturbation expansion is different in that the usual freedom in choosing
the Green's functions (like working with the advanced or retarded Green's functions or the
Feynman propagator) is removed. The expansion becomes unique by combining
causality with suitable normalization conditions for the states of the fermionic projector.

The paper is organized as follows. In Section~\ref{sec2}, we give a self-contained introduction to the construction of the generalized Dirac sea~\eqref{Pint} in terms of a formal perturbation expansion in $\mathscr{B}$. 
In Section~\ref{sec3}, we explain the rescaling procedure for the states of the interacting Dirac sea, thus obtaining a unique idempotent operator in terms of a formal perturbation expansion, the so-called {\em{rescaled causal perturbation
expansion}}. The main technical task is to elaborate the combinatorics of the different expansions in detail; this will be carried out in Section~\ref{sec4}. As explained in Remark~\ref{remark}, the rescaling formally reproduces results from spectral theory and functional analysis in the setting of perturbation expansions.
An interesting consequence is the existence of unitary transformations between the free and the generalized Dirac seas. We prove this fact in Section~\ref{sec5} by deriving and analyzing equations for the perturbation flow. Finally, in Section~\ref{sec6} it is shown that the rescaling procedure has no influence on the residual argument and the light-cone expansion as worked out in~\cite{light}. But it does change the form of
the so-called high-energy contribution.

\section{The Causal Perturbation Expansion} \label{sec2}
In this section we give a self-contained review of the causal perturbation expansion as developed
in~\cite{sea} (see also~\cite[Chapter 2]{PFP}). The rescaling procedure will then be explained
in Section~\ref{sec3}. We always assume that the mass is a positive parameter,
\[ m>0\:. \]
We decompose the Fourier integral~\eqref{Pvac} as follows,
\begin{align}
P^{sea}_m(x,y) &= \frac{1}{2}\Big(p_m(x,y)-k_m(x,y)\Big) \,, \label{sea-pk} \\
\intertext{where}
\begin{split}
p_m(x,y)&=\int\frac{d^4q}{(2\pi)^4}\:(\slashed{q}+m)\:\delta(q^2-m^2)\:e^{-iq(x-y)} \\
k_m(x,y)&=\int\frac{d^4q}{(2\pi)^4}\:(\slashed{q}+m)\:\delta(q^2-m^2)\:\epsilon(q_0)\:e^{-iq(x-y)}\:.
\end{split} \label{def-p-k}
\end{align}
Here~$\slashed{q} \equiv q_\mu \gamma^\mu$, $\epsilon$ is the sign function,
and the function~$\Theta$ in~(\ref{Pvac}) is the Heaviside function. For the signature of the Minkowski inner product we use the convention $(+ - -\, -)$.
All these Fourier integrals are well-defined tempered distributions. The splitting~\eqref{sea-pk}
gives rise to the decomposition~\eqref{Psea} of the corresponding operators.

The decomposition~\eqref{sea-pk} reveals the following connection to causality.
The Dirac equation is causal in the sense that information propagates at most with the speed of light.
This is reflected in a support property of the advanced and retarded Green's
functions, which we denote by~$s_m^\vee$ and~$s_m^\wedge$, respectively.
They have the Fourier representation
\begin{align*}
s_m^{\vee}(x,y) &=\int\frac{d^4q}{(2\pi)^4}\:\frac{\slashed{q}+m}{q^2-m^2-i\varepsilon q_0}\:e^{-iq(x-y)} \\
s_m^{\wedge}(x,y) &=\int\frac{d^4q}{(2\pi)^4}\:\frac{\slashed{q}+m}{q^2-m^2+i\varepsilon q_0}\:e^{-iq(x-y)} \:,
\end{align*}
where $\varepsilon>0$ is a regularization parameter, and it is understood implicitly
that one should take the limit $\varepsilon \searrow 0$ in the distributional sense.
Computing the integrals with residues, one readily verifies that~$\supp(s_m^\vee(x,.))\subset J_x^\vee$,
where
\[ J_x^\vee = \{y \text{ with } (y-x)^2 \geq 0 \text{ and } y^0>x^0\} \]
denotes the future light cone centered 
at~$x$ (similarly, $\supp(s_m^\wedge(x,.))\subset J_x^\wedge$). Taking the
difference of the two expressions and using the identity
\begin{equation}
\delta(x)=\frac{1}{2\pi i}\left(\frac{1}{x-i\varepsilon}-\frac{1}{x+i\varepsilon}\right),
	\label{eq:delta-formula}
\end{equation}
one finds that
\begin{equation}\label{k-ss}
\boxed{	\quad k_m=\frac{1}{2\pi i}(s_m^{\vee}-s_m^{\wedge})\:.  \quad}
\end{equation}
This relation shows that~$k_m$ is a causal operator in the sense that $\supp(k_m(x,.))\subset J_x$,
where $J_x =J_x^\vee \cup J_x^\wedge$ is the light cone centered at $x$.

We point out that the operator $p_m$ is not causal in the above sense.
To see this, we decompose~$p_m$ in analogy to~(\ref{k-ss}) as
\begin{equation}
	p_m=\frac{1}{2\pi i}(s_m^+-s_m^-) \:,
	\label{p-ss}
\end{equation}
where the operators~$s^\pm$ have the Fourier representation
\beq \label{feynman}
s_m^\pm(x,y) = \int\frac{d^4q}{(2\pi)^4}\:\frac{\slashed{q}+m}{q^2-m^2 \mp
i\varepsilon}\:e^{-iq(x-y)} \:.
\eeq
The operator~$s_m^-$ is known in the literature as the Feynman propagator, characterized
by the condition that positive-energy solutions propagate forward in time while negative-energy solutions propagate backwards in time. An explicit calculation of the Fourier integral~\eqref{feynman}
in terms of Bessel functions shows that~$s^+_m$ and $s^-_m$ as well as their difference
do not vanish outside the light cone and are thus not causal.

For the subsequent constructions it is important to observe that the operator~$p_m$ can be
obtained from~$k_m$ as follows. Defining the absolute value of a diagonalizable matrix
$A$ as the unique positive semi-definite matrix $|A|$ with $A^2=|A|^2$, we find
that $|\epsilon(k^0)(\slashed{k}+m)|=(\slashed{k}+m)$. Since the operators~$p_m$ and~$k_m$
are diagonal in momentum space, taking the absolute value pointwise in momentum space can
be understood formally as taking the absolute value of the corresponding operator acting
on the Dirac wave functions,
\begin{equation}\label{p-absk}
\boxed{ \quad p_m=|k_m| \:. \quad }
\end{equation}

In the remainder of this section we shall generalize the relations~\eqref{k-ss} and~\eqref{p-absk}
to the case with general interaction~\eqref{tdep}; in Section~\ref{sec3} we will then develop
a method for generalizing the definition of the fermionic projector~\eqref{Psea}. \\[-0.8em]

Using the causal support property, the advanced and retarded Green's functions~$\tilde{s}_m^{\vee}$
and~$\tilde{s}_m^{\wedge}$
are uniquely defined even in the case with interaction~\eqref{tdep}.
This could be done non-perturbatively using the theory of symmetric hyperbolic
systems~\cite{john}. For our purpose, it is sufficient to give the unique perturbation series
\beq \tilde{s}_m^{\vee}=\sum_{n=0}^{\infty}(-s_m^{\vee}\mathscr{B})^ns_m^{\vee}\;, \qquad
\tilde{s}_m^{\wedge}=\sum_{n=0}^{\infty}(-s_m^{\wedge}\mathscr{B})^ns_m^{\wedge}\:.
\label{series-scaustilde}
\eeq
Here the operator products involving  the potential $\mathscr{B}$ are defined as follows,
\begin{equation}
	(s_m^{\vee}\mathscr{B}s_m^{\vee})(x,y):=\int d^4z \:s_m^{\vee}(x,z)\mathscr{B}(z)s_m^{\vee}(z,y).
	\label{eq:defof-bprods}
\end{equation}
It is straightforward to verify that the perturbation series~(\ref{series-scaustilde})
indeed satisfy the defining relations for the advanced Green's function
\[ (i \slashed{\partial}_x + {\mathscr{B}}(x) - m)\: \tilde{s}^{\vee}_m(x,y) = \delta^4(x-y) \:,\quad
\supp s_m^{\vee}(x,.) \subset J_x^\vee \]
(and similarly for the retarded Green's function; for details see~\cite[\S2.2]{PFP}).

Moreover, we define the operators
\begin{align}\label{def-s}
	s_m&=\frac{1}{2}(s_m^{\vee}+s_m^{\wedge})=\frac{1}{2}(s_m^{+}+s_m^{-}),\\	
	F_m(Q,n)&=\begin{cases}p_m,\:\:n\in Q\\
	k_m,\:\:n\notin Q\end{cases}
	\hspace{1cm} \mbox{for $Q\subset\mathbb{N}$, $n\in \mathbb{N}$}
	\label{def-F}
\end{align}
and the series of operator products
\begin{align}
 b_m^<=\sum_{n=0}^{\infty}(-s_m\mathscr{B})^n\;,\hspace{1cm}  b_m=\sum_{n=0}^{\infty}(-\mathscr{B}s_m)^n\mathscr{B}\;,\hspace{1cm} b_m^>=\sum_{n=0}^{\infty}(-\mathscr{B}s_m)^n\:.
	\label{thm-defs}
\end{align}

\begin{Lemma} \label{lemma21}
The following identities hold:
\begin{align}
p_m\,p_{m'}&=k_m\,k_{m'}=\delta(m-m')\:p_m\label{eq:pp-p} \\
p_m\,k_{m'}&=k_m\,p_{m'}=\delta(m-m')\:k_m\label{eq:pk-k} \\
p_m \,s_{m'}&=s_{m'}\,p_m=\frac{\pv}{m-m'}\:p_m\label{eq:ps-p} \\
k_m \,s_{m'}&=s_{m'}\,k_m=\frac{\pv}{m-m'}\:k_m\label{eq:ks-k} \\
s_m\,s_{m'}&=\frac{\pv}{m-m'}\:(s_m-s_{m'})+\pi^2\delta(m-m')\:p_m\;,\label{eq:ss-sp}
\end{align}
where the principle value distribution is given by $\frac{\pv}{x}:=\frac{1}{2}[(x+i\varepsilon)^{-1}+(x-i\varepsilon)^{-1}]$.
\end{Lemma}
\Proof
Calculating pointwise in momentum space, we obtain 
\begin{align*}	p_m(q)\,p_{m'}(q)&=\delta(q^2-m^2)(\slashed{q}+m)\delta(q^2-m'^2)(\slashed{q}+m')\\
&=\delta(m^2-m'^2)\delta(q^2-m^2) \Big(q^2+(m+m')\slashed{q}+mm' \Big)\\
&=\frac{1}{2m}\delta(m-m')\delta(q^2-m^2)(m^2+(m+m')\slashed{q}+mm')\\
&=\frac{1}{2m}\delta(m-m')\delta(q^2-m^2)2m(m+\slashed{q})=\delta(m-m')\,p_m(q) \:.
\end{align*}
This gives the first part of \eqref{eq:pp-p}. The second part of this formula as well as formula \eqref{eq:pk-k} are obtained analogously. The formulas \eqref{eq:ps-p} and \eqref{eq:ks-k} are computed in the following manner:
\begin{align*}	&2 p_m(q)\,s_{m'}(q)=\delta(q^2-m^2)(\slashed{q}+m) \left(\frac{\slashed{q}+m'}{q^2-m'^2-i\varepsilon q_0}+\frac{\slashed{q}+m'}{q^2-m'^2+i\varepsilon q_0}\right)\\
&=\delta(q^2-m^2)(q^2+(m+m')\slashed{q}+mm') \left(\frac{1}{q^2-m'^2-i\varepsilon q_0}+\frac{1}{q^2-m'^2+i\varepsilon q_0}\right)\\
&=\delta(q^2-m^2)(m^2+(m+m')\slashed{q}+mm') \left(\frac{1}{m^2-m'^2-i\varepsilon q_0}+\frac{1}{m^2-m'^2+i\varepsilon q_0}\right)\\
&=\delta(q^2-m^2)(\slashed{q}+m) \left(\frac{(m+m')}{(m+m')(m-m')-i\varepsilon q_0}+\frac{(m+m')}{(m+m')(m-m')+i\varepsilon q_0}\right)\\
&=2 \frac{\pv}{m-m'}\,p_m(q)\:.
\end{align*}

The derivation of~\eqref{eq:ss-sp} is a bit more difficult. By \eqref{k-ss} and \eqref{def-s}, we have
\begin{align*}
	s_m=s_m^{\vee}-i\pi k_m=s_m^{\wedge}+i\pi k_m\,.
\end{align*}
Thus we can express the product $s_m(q)\,s_{m'}(q)$ in two ways, namely as
\begin{align*}
	s_m(q)\,s_{m'}(q)=&(s_m^{\vee}(q)-i\pi k_m(q))(s_{m'}^{\vee}(q)-i\pi k_{m'}(q))\\
	=&s_m^{\vee}(q)\,s_{m'}^{\vee}(q)-i\pi\left(k_{m'}(q)\frac{1}{m'-m-i\varepsilon q_0}+k_{m}(q)\frac{1}{m-m'-i\varepsilon q_0}\right)\\&\;\;\;\;-\pi^2\delta(m-m')\,p_m(q)\:,
\end{align*}
or alternatively as
\begin{align*}
	s_m(q)\,s_{m'}(q)=&(s_m^{\wedge}(q)+i\pi k_m(q))(s_{m'}^{\wedge}(q)+i\pi k_{m'}(q))\\
	=&s_m^{\wedge}(q)\,s_{m'}^{\wedge}(q)+i\pi\left(k_{m'}(q)\frac{1}{m'-m+i\varepsilon q_0}+k_{m}(q)\frac{1}{m-m'+i\varepsilon q_0}\right)\\&\;\;\;\;-\pi^2\delta(m-m')\,p_m(q)\,.
\end{align*}
Adding these two formulas yields
\begin{eqnarray*}
\lefteqn{ 2s_m(q)\,s_{m'}(q) - (s_m^{\vee}(q)\,s_{m'}^{\vee}(q)+s_m^{\wedge}(q)\,s_{m'}^{\wedge}(q))
+2\pi^2\delta(m-m')\,p_m(q) } \\
&=& i\pi k_{m'}(q)\left(\frac{1}{m'-m+i\varepsilon q_0}-\frac{1}{m'-m-i\varepsilon q_0}\right)\\
&&+i\pi k_{m}(q)\left(\frac{1}{m-m'+i\varepsilon q_0}-\frac{1}{m-m'-i\varepsilon q_0}\right)\\
&\overset{\eqref{eq:delta-formula}}{=}& i\pi k_{m'}(q)\epsilon(-q_0)\,2\pi i\,\delta(m'-m)
+i\pi k_{m}(q)\epsilon(-q_0)\,2\pi i\,\delta(m-m')\\
&=&-2\pi^2\delta(m'-m)(-p_{m'}(q)) - 2\pi^2\delta(m-m')(-p_{m}(q))\:,
\end{eqnarray*}
where in the last step we used the definitions of $p_m$ and $k_m$. We thus obtain
\beq 
s_m \,s_{m'} =  \frac{1}{2}
\left( s_m^{\vee}\, s_{m'}^{\vee} +s_m^{\wedge} \,s_{m'}^{\wedge} \right)
+ \pi^2\delta(m-m')\, p_m\:. \label{2ssprime}
\eeq

It remains to derive the relations
\beq \label{smmp}
s_m^{\vee}\,s_{m'}^{\vee}=\frac{\pv}{m-m'}(s_m^{\vee}-s_{m'}^{\vee})
\qquad \text{and} \qquad
s_m^{\wedge}\,s_{m'}^{\wedge}=\frac{\pv}{m-m'}(s_m^{\wedge}-s_{m'}^{\wedge})\:,
\eeq
which can be regarded as ``resolvent identities'' for the causal Green's functions.
It suffices to consider the case of the advanced Green's function.
Clearly, the operators on the right side of~\eqref{smmp}
satisfy the support condition~$\supp((s_m^{\vee}-s_{m'}^{\vee})(x,.))\subset J_x^{\vee}$, and from
\begin{equation*}
	s_m^{\vee}\,s_{m'}^{\vee}(x,y)=\int d^4z\: s_m^{\vee}(x,z)\,s_{m'}^{\vee}(z,y)
\end{equation*}
we see that the operators on the left side of~(\ref{smmp}) satisfy this support condition as well.
Moreover, the calculations
$$(i\slashed{\partial}_x-m)\:s_m^{\vee}\,s_{m'}^{\vee}(x,y)=s_{m'}^{\vee}(x,y) $$
and
\begin{align*}
(i\slashed{\partial}_x&-m)\frac{\pv}{m-m'}(s_m^{\vee}-s_{m'}^{\vee})(x,y)\\
&=\frac{\pv}{m-m'}[\delta(x-y)-(m'-m)s_{m'}^{\vee}(x,y)-\delta(x-y)] =s_{m'}^{\vee}(x,y)
\end{align*}
show that both sides of~\eqref{smmp} satisfy the same inhomogeneous Dirac equation.
Hence their difference is a distributional solution of the homogeneous Dirac equation
vanishing outside~$J_x^\vee$. The uniqueness of the solution of the Cauchy problem
for hyperbolic PDEs yields that this difference vanishes identically.
This shows~\eqref{smmp} and thus finishes the proof of~\eqref{eq:ss-sp}.
\end{proof}
We remark that the last summand in~\eqref{eq:ss-sp} was by mistake omitted in~\cite{PFP}.
For this reason, we shall now rederive and correct the formulas of the causal perturbation expansion
in detail.

\begin{Corollary} \label{corol22}
Let $C\in \{p_m,k_m\}$ and $C^{\,\prime}\in\{p_{m'},k_{m'}\}$ and~$b_m^<$, $b_m^>$
as in~\eqref{thm-defs}. Then the following calculation rule holds:
\begin{align}
	C\,b_m^>b_{m'}^<\,C^{\,\prime}=CC^{\,\prime}+\delta(m-m')\:
	\pi^2C\,b_mp_mb_m\,C^{\,\prime}.
\end{align}
\end{Corollary}
\begin{proof}
Using the calculation rules of the previous lemma, we obtain
\[ C\Big(\sum_{l=0}^{1}(\mathscr{B}s_m)^l(s_{m'}\mathscr{B})^{n-l}\Big)C^{\,\prime}
= C s_{m'}\mathscr{B} C^{\,\prime} + C \mathscr{B} s_{m} C^{\,\prime}
= \frac{\pv}{m-m'} \left( C \mathscr{B} C^{\,\prime} - C \mathscr{B} C^{\,\prime} \right) = 0\:. \]
The same method also applies to higher order. We then get a telescopic sum,
but the last summand in~\eqref{eq:ss-sp} gives additional contributions. More precisely,
for any $n\geq2$,
\begin{align*}	
&C\Big(\sum_{l=0}^{n}(\mathscr{B}s_m)^l(s_{m'}\mathscr{B})^{n-l}\Big)C^{\,\prime}=\\
&=C(\mathscr{B}s_m)^nC^{\,\prime}+C(s_{m'}\mathscr{B})^nC^{\,\prime}+C\Big[\sum_{l=1}^{n-1}(\mathscr{B}s_m)^l(s_{m'}\mathscr{B})^{n-l}\Big]C^{\,\prime}\\
&=C\,\frac{\pv}{m-m'} \left[ -(\mathscr{B}s_m)^{n-1}\mathscr{B}+\mathscr{B}(s_{m'}\mathscr{B})^{n-1} \right]
C^{\,\prime}\\&\;\;\;+C\sum_{l=1}^{n-1}(\mathscr{B}s_m)^{l-1}\mathscr{B}\left(\frac{\pv}{m-m'}(s_m-s_{m'})+\delta(m-m')\pi^2p_m\right)\mathscr{B}(s_{m'}\mathscr{B})^{n-l-1}C^{\,\prime}\\
&=\frac{\pv}{m-m'}\,C \left[-(\mathscr{B}s_m)^{n-1}\mathscr{B}+\mathscr{B}(s_{m'}\mathscr{B})^{n-1} \right] C^{\,\prime}\\
&\;\;\;+\frac{\pv}{m-m'}\,C\Big(\sum_{l=1}^{n-1}(\mathscr{B}s_m)^{l}(\mathscr{B}s_{m'})^{n-l-1}\mathscr{B}-
\sum_{l=0}^{n-2}(\mathscr{B}s_m)^{l}(\mathscr{B}s_{m'})^{n-l-1}\mathscr{B}\Big)C^{\,\prime}\\
&\;\;\;+\delta(m-m')\pi^2\,C\sum_{l=1}^{n-1}(\mathscr{B}s_m)^{l-1}\mathscr{B}p_m\mathscr{B}(s_{m'}\mathscr{B})^{n-l-1}C^{\,\prime}\\
&=\delta(m-m')\pi^2\,C\sum_{l=1}^{n-1}(\mathscr{B}s_m)^{l-1}\mathscr{B}p_m\mathscr{B}(s_{m'}\mathscr{B})^{n-l-1}C^{\,\prime}\\
&=\delta(m-m')\pi^2\,C\sum_{l=0}^{n-2}(\mathscr{B}s_m)^{l}\mathscr{B}p_m\mathscr{B}(s_{m'}\mathscr{B})^{n-l-2}C^{\,\prime}.
\end{align*}
We thus obtain
\begin{align*}	C\,b_m^>b_{m'}^<\,C^{\,\prime}&=C\sum_{n=0}^{\infty}(-\mathscr{B}s_m)^n\sum_{n'=0}^{\infty}(-s_{m'}\mathscr{B})^{n'}C^{\,\prime}\\
&=C\sum_{n=0}^{\infty}\sum_{l=0}^{n}(-\mathscr{B}s_m)^l(-s_{m'}\mathscr{B})^{n-l}C^{\,\prime}\\
&=CC^{\,\prime}+\delta(m-m')\pi^2\sum_{n=2}^{\infty}(-1)^nC\left(\sum_{l=0}^{n-2}(\mathscr{B}s_m)^{l}\mathscr{B}p_m\mathscr{B}(s_{m'}\mathscr{B})^{n-l-2}\right)C^{\,\prime}\\
&=CC^{\,\prime}+\delta(m-m')\pi^2\sum_{n=0}^{\infty}(-1)^nC\left(\sum_{l=0}^{n}(\mathscr{B}s_m)^{l}\mathscr{B}p_m\mathscr{B}(s_{m'}\mathscr{B})^{n-l}\right)C^{\,\prime}\\
&=CC^{\,\prime}+\delta(m-m')\pi^2C\,b_mp_mb_m\,C^{\,\prime}.
\end{align*} \vspace*{-1cm}

\end{proof}

We next generalize the operator $k_m$ to the case of non-static external potentials. The natural way to do this is to define
\beq
\tilde{k}_m:=\frac{1}{2\pi i}(\tilde{s}_m^{\vee}-\tilde{s}_m^{\wedge})\;,   \label{def-ktil}
\eeq
giving a rigorous non-perturbative definition of $\tilde{k}_m$. The formal perturbation expansion of $\tilde{k}_m$ is given in the following proposition.

\begin{Prp}\label{thm:kpert}
The relation \eqref{def-ktil}  uniquely determines the perturbation expansion for $\tilde{k}_m$. We have the formula

\begin{equation}\label{k-pert}
	\tilde{k}_m=\sum_{\beta=0}^{\infty}(-i\pi)^{2\beta}b_m^<k_m(b_mk_m)^{2\beta}b_m^>\:.
\end{equation}
\end{Prp}

\begin{proof}[Proof]
An explicit calculation shows that
\beq \label{brel}
(i\slashed{\partial}+\mathscr{B}-m)\, b_m^< = 0\:.
\eeq
As all operator products in \eqref{k-pert} have a factor $b_m^<$ at the left, the series in \eqref{k-pert} is a solution of the Dirac equation.

From \eqref{k-ss} and \eqref{def-s}, we have
\beq s_m^{\vee} = s_m+i\pi k_m\:, \qquad s_m^{\wedge} = s_m-i\pi k_m\:. \label{svw-sk}
\eeq
We substitute the series \eqref{series-scaustilde} into \eqref{def-ktil}, insert~\eqref{svw-sk} and expand. A reordering of the resulting sum gives the claim. The details of the reordering process can be found
in~\cite{sea}.
\end{proof} \noindent
In Appendix~\ref{appA} we give a compilation of the perturbation series of~$\tilde{k}_m$
and of other perturbation series which will appear in what follows, 
up to third order in~${\mathscr{B}}$.

Next, we want to generalize the operator $p_m$ to the case with interaction.
Note that for the construction of $\tilde{k}_m$, it was crucial that $k_m$ was a causal operator. However, $p_m$ does not have this property, and it is
not at all obvious how to generalize $p_m$ using the decomposition \eqref{p-ss}. Namely,
as mentioned after~\eqref{tdep}, in the time-dependent setting
the energy is not a conserved quantity, and thus the frequency conditions characterizing
the Feynman propagator (see after~\eqref{feynman}) have no meaningful generalization.
Instead, we want to exploit equation \eqref{p-absk} to generalize the operator $p_m$, i.e. 
\begin{equation}
	\tilde{p}_m \stackrel{\text{formally}}{:=}\sqrt{\tilde{k}_m^2}\:.
	\label{ptil-formally}
\end{equation}
This will be done using a Taylor series for the square root function, defining~$\tilde{p}_m$
by a perturbation series.

Using Corollary 2.2 and Proposition 2.3, the operator $\tilde{k}_m\tilde{k}_{m'}$ is calculated to be
\begin{align}\nonumber	
\tilde{k}_m& \tilde{k}_{m'} = \delta(m-m') \\
&\times\sum_{\beta_1,\beta_2=0}^{\infty}(-i\pi)^{2\beta_1+2\beta_2}b_m^<(k_mb_m)^{2\beta_1}\Big(p_m+\pi^2k_mb_mp_mb_mk_m\Big)(b_mk_m)^{2\beta_2}b_m^> \:.
	\label{eq:}
\end{align}
In the following calculations, all operator products can be computed using the rules of
Lemma~\ref{lemma21} and Corollary~\ref{corol22}. This always yields a factor $\delta(m-m')$,
and thus for notational simplicity we can omit these factors $\delta(m-m')$ \label{deltaweg}
and consider all expressions at the same value of $m$. Furthermore, leaving out the subscripts `$_m$', our calculation rules can be
written in the compact form
\begin{align}
\label{eq:rule-pp}	p^2&=k^2=p \\
\label{eq:rule-pk}	pk&=kp=k \\
\label{eq:rule-bb}	b^>b^<&=p+\pi^2pbpbp\,.
\end{align}
In order to keep the combinatorics simple, it is convenient to first treat the
summand~$\pi^2pbpbp$ in~\eqref{eq:rule-bb} by rearranging the
perturbation series of $\tilde{k}$.

\begin{Def} \label{def-b_cech} We define the series of operator products
\begin{align}	\breve{b}^<&:=b^<\,(p+\pi^2pbpbp)^{-1/2} \;=\; b^<\,\Big(p+\sum_{n\geq1}(-1)^n\frac{(2n-1)!!}{2^nn!}(\pi^2pbpbp)^n\Big)\label{eq:b-cechauf} \\
\breve{b}^>&:= (p+\pi^2pbpbp)^{-1/2}\,b^> \;=\; \Big(p+\sum_{n\geq1}(-1)^n\frac{(2n-1)!!}{2^nn!}(\pi^2pbpbp)^n\Big)\,b^>\label{eq:b-cechzu} \\
p+A&:=(p+\pi^2pbpbp)^{1/2} \;=\; p+\sum_{n\geq1}(-1)^{n+1}\frac{(2n-3)!!}{n!2^n}(\pi^2pbpbp)^n \:.
	\label{eq:pplusA}
\end{align}
\end{Def}
Manipulating formal power series, like for example
\begin{align*}
\breve{b}^>\breve{b}^<&=(p+\pi^2pbpbp)^{-1/2}b^>b^<(p+\pi^2pbpbp)^{-1/2}\\
&=(p+\pi^2pbpbp)^{-1/2}(p+\pi^2pbpbp)(p+\pi^2pbpbp)^{-1/2}=p\:,
\end{align*}
we find the simpler calculation rules
\begin{align}
	\breve{b}^>\breve{b}^<&=p \label{first-simplerule}\:, \\
	\breve{b}^<(p+A)&=b^<\label{2nd-simplerule}\:,\\
	(p+A)\breve{b}^>&=b^>\label{3rd-simplerule}\:,\\
	(p+A)^2&=p+\pi^2pbpbp \:.\label{last-simplerule}
\end{align}
Thus we can rewrite $\tilde{k}$ and $\tilde{k}^2$ as
\begin{align}\label{ktil-rewritten}	\tilde{k}&=\sum_{\beta=0}^{\infty}(-i\pi)^{2\beta}\:\breve{b}^<(p+A)k(bk)^{2\beta}(p+A)\breve{b}^>,\\
\tilde{k}^2&=\sum_{\beta_1,\beta_2=0}^{\infty}(-i\pi)^{2(\beta_1+\beta_2)}\:\breve{b}^<(p+A)k(bk)^{2\beta_1}(p+\pi^2pbpbp)k(bk)^{2\beta_2}(p+A)\breve{b}^>.\label{ktil2-rewritten}
\end{align}
We are now ready to derive the perturbation expansion for $\tilde{p}$.
\begin{Thm}
The relation \eqref{ptil-formally}  uniquely determines the perturbation expansion for $\tilde{p}$. We have the formula
\begin{align}\label{ptil-pert}
\tilde{p}=\breve{b}^<\left(p+\sum_{n\geq1}c_nX^n\right)\breve{b}^>
\end{align}
with the combinatorial factors
\begin{align}
	c_n=(-1)^{n+1}\frac{(2n-3)!!}{n!2^n}
\end{align}
and the operator
\beq \label{X-1stversion}
X=-p+\sum_{\beta_1,\beta_2=0}^{\infty}(-i\pi)^{2(\beta_1+\beta_2)}(p+A)k(bk)^{2\beta_1}(p+A)^2k(bk)^{2\beta_2}(p+A)\:.
\eeq
\end{Thm}
\Proof Using the calculation rule~\eqref{first-simplerule}, all the intermediate factors~$\breve{b}^<$
and~$\breve{b}^>$ will always drop out of our calculations. Therefore, for ease in notation
we can simply omit all factors~$\breve{b}^<$ and~$\breve{b}^>$.
Following~\eqref{ptil-formally}, we are thus looking for a positive operator $\tilde{p}$ being a  powers series in $\mathscr{B}$, such that
$$\tilde{p}^2=\tilde{k}^2\stackrel{\eqref{ktil2-rewritten}}{=}\sum_{\beta_1,\beta_2=0}^{\infty}(-i\pi)^{2(\beta_1+\beta_2)}(p+A)k(bk)^{2\beta_1}(p+\pi^2pbpbp)k(bk)^{2\beta_2}(p+A) \:.$$
Using again the operator $X$ defined in~\eqref{X-1stversion}, we have
\begin{equation}
	\tilde{p}^2=p+X.
	\label{p2-pX}
\end{equation}
The operator $p$ is idempotent and acts as the identity on $X$. Thus we can take the square root of $p+X$ with a formal Taylor expansion,
\begin{align}\label{ptil-taylor}
	\tilde{p}=\sqrt{p+X}=p+\sum_{n\geq1}(-1)^{n+1}\frac{(2n-3)!!}{n!2^n}X^n \:,
\end{align}
which uniquely defines $\tilde{p}$ as a positive operator. Reinserting the factors $\breve{b}^<$, $\breve{b}^>$ into \eqref{ptil-taylor}, we obtain the formula \eqref{ptil-pert}.
\QED
Note that, due to~\eqref{eq:b-cechauf} and~\eqref{brel}, the series \eqref{ptil-pert} is a solution of
the Dirac equation~\eqref{tdep}. Furthermore, rewriting~\eqref{X-1stversion} as
\begin{align}
X &=2A+A^2+\pi^2(kbpbk+Akbpbk+kbpbkA+AkbpbkA) \nonumber \\
&\;\;\;+\sum_{(\beta_1,\beta_2)\neq(0,0)}(-i\pi)^{2(\beta_1+\beta_2)}(p+A)k(bk)^{2\beta_1}(p+\pi^2pbpbp)k(bk)^{2\beta_2}(p+A)\:, \label{X-2ndversion}
\end{align}
one sees that~$X$ only involves operator products with at least two powers of $b$ (and thus also at least two powers of $\mathscr{B}$). Consequently, $X^n$ is of order $\O(\mathscr{B}^{2n})$,
so that~\eqref{ptil-pert} is indeed a well-defined power series in $\mathscr{B}$.

The following Lemma shows that under suitable regularity and decay assumptions on the potential $\mathscr{B}$, all operator products appearing in~\eqref{k-pert} and~\eqref{ptil-pert} are well-defined
and finite.

\begin{Lemma}\label{regular-lemma}
Let $n\in\mathbb{N}$ and $C_0,C_1, \ldots,C_n\in\{k,p,s\}$. If the external potential is smooth and decays so fast at infinity that the functions $\mathscr{B}(x)$, $x_i\mathscr{B}(x)$ and $x_ix_j\mathscr{B}(x)$ are integrable, then the operator product
\begin{align}
	(C_n\mathscr{B}C_{n-1}\mathscr{B}\ldots\mathscr{B}C_0)(x,y)
\end{align}
is a well-defined tempered distribution on $\mathbb{R}^4\times\mathbb{R}^4$.
\end{Lemma} \noindent
For the proof, one derives Schwarz norm bounds by combining estimates in position and
in momentum space. The details can be found in \cite[Lemma~2.2.2]{PFP}.

\section{The Rescaling Procedure} \label{sec3}
Introducing in analogy to~\eqref{sea-pk} the operator
\begin{align}
	\tilde{t}:=\frac{1}{2}(\tilde{p}-\tilde{k}) \:,
\label{eq:tt}
\end{align}
the key feature of the operator $\tilde{t}$ is that it specifies via its range a certain subspace of the space of solutions of the Dirac equation \eqref{tdep}. Since in the vacuum case, the range of~$\tilde{t}$
is precisely the space of all negative-energy solutions, we refer to the range of $\tilde{t}$ as the
\textit{generalized negative-energy solutions}.

In~\cite{sea, PFP}, the interacting fermionic projector is defined by~\eqref{eq:tt}. However,
the operator~$\tilde{t}$ is {\em{not}} a projection operator in the sense of \eqref{P-proj}.
This can already be seen by considering the perturbation expansions to second order.
Namely, from~\eqref{ktil-rewritten} and~\eqref{ptil-pert},
\[ \tilde{t}=\frac{1}{2} \: \breve{b}^< \left(
p-k+kbkbk+\frac{1}{2}kbpbk+\frac{1}{2}pbpbp-\frac{1}{2}pbkbk-\frac{1}{2}kbkbp \right)
\breve{b}^> +\O(b^4) \]
Taking the square and using the calculation rules (\ref{first-simplerule}-\ref{last-simplerule}),
we obtain
\[ \begin{split}
\tilde{t}\:^2 = \frac{1}{2}\: \breve{b}^< \Big( p-k +\frac{3}{2}kbkbk&+\frac{1}{2}kbpbk+\frac{1}{2}pbpbp-pbkbk \\
&-kbkbp+\frac{1}{2}pbkbp -\frac{1}{2}kbpbp-\frac{1}{2}pbpbk \Big) \breve{b}^>  +\O(b^4)\:.
\end{split} \]
Obviously, these expansions for~$\tilde{t}$ and~$\tilde{t}^2$ do not coincide.

Our strategy for resolving this problem is to rescale the states in the range of~$\tilde{t}$
with the following general procedure. Since the fermionic projector~$P$ should have the
same range as the operator~$\tilde{t}$, we take the ansatz
\[	P=\tilde{t}Y\tilde{t} \]
with a ``rescaling operator'' $Y$ which should be invertible and should commute with~$\tilde{t}$.
We define
\beq \label{defRZ}
R = \tilde{t}\tilde{t} \qquad \text{and} \qquad
Z = \tilde{k}\tilde{k}=\tilde{p}\tilde{p}\:.
\eeq
The simplest idea would be to choose ``$Y=R^{-1}$'', but as the range of $\tilde{t}$ 
are only the generalized negative-energy solutions, the operator~$R$ will in general not be invertible. On the other hand, the operator $\tilde{k}$ maps onto the whole space of solutions of the Dirac equation, and thus the formal definition ``$Y=Z^{-1}$'' does the trick.
We will now make this definition precise and show that it indeed yields a projection operator.
As explained on page~\pageref{deltaweg}, we again omit all factors $\breve{b}^<$ and~$\breve{b}^>$, knowing that they can
be reinserted at the very end.
According to \eqref{p2-pX}, $Z$ is given by $Z=p+X$.
As $p$ is idempotent and acts as the identity on the operator $X$, we can define~$Y$
by a Neumann series,
\begin{equation}\label{def-x}
	Y=p+\sum_{l\geq1}(-1)^lX^l\:.
\end{equation}
 Then
 \begin{align}
	ZY=(p+X)(p+\sum_{l\geq1}(-X)^l)=p\,,
\end{align}
showing that $Z$ and $X$ are inverse in our context where $p$ acts as the identity on all other operators.

We are now ready to define the fermionic projector in the interacting case~\eqref{tdep} and to prove that we indeed obtain a projection operator.

\begin{Def}
The fermionic projector~$P$ is defined by 
\begin{equation}\label{eq:defP}
	P:=\tilde{t}Y\tilde{t} \:.
\end{equation}
with $\tilde{t}$ according to~\eqref{eq:tt} and~$Y$ given by~\eqref{def-x}.
\end{Def}

\begin{Prp}
The fermionic projector is idempotent,
$$P^2=P \:,$$
(where, as explained on page~\pageref{deltaweg}, we leave out the factor~$\delta(m-m')$ and compute
for a fixed mass parameter~$m$).
\end{Prp}

\begin{proof}[Proof]
The formulas \eqref{ktil-rewritten} and \eqref{X-1stversion} show that the operators $\tilde{k}$ and $X$ commute,
\begin{align*}
\tilde{k} & (p+X) = (p+X)\tilde{k} \\
&= \sum_{\substack{\beta_1,\beta_2,\\\beta_3=0}}^{\infty}(-i\pi)^{2(\beta_1+\beta_2+\beta_3)}(p+A)k(bk)^{2\beta_1}(p+A)^2k(bk)^{2\beta_2}(p+A)^2k(bk)^{2\beta_3}(p+A) \:.
\end{align*}
Noting that the operators~$\tilde{p}$ and~$Y$ are formed as a sum of powers of~$X$
(cf.~\eqref{ptil-taylor} and~\eqref{def-x}), it follows that these operators both
commute with each other and with~$\tilde{k}$,
$$[\tilde{p},Y] = [\tilde{p},\tilde{k}] = [Y,\tilde{k}] =0 \:.$$
Hence
\begin{align}
	\tilde{t}Y\tilde{t}=\frac{1}{4}(\tilde{p}-\tilde{k})Y(\tilde{p}-\tilde{k})=
	\frac{1}{4}(\tilde{p}\tilde{p}Y+\tilde{k}\tilde{k}Y-\tilde{p}Y\tilde{k}-\tilde{k}Y\tilde{p})=
	\frac{1}{2}(p-\tilde{p}Y\tilde{k}) \:,\label{Psea-start}
\end{align}
and thus
\begin{align*}
PP&=(\tilde{t}Y\tilde{t})(\tilde{t}Y\tilde{t})=\frac{1}{4}(p-\tilde{p}Y\tilde{k})(p-\tilde{p}Y\tilde{k})\\
&=\frac{1}{4}(p-2\tilde{p}Y\tilde{k}+\tilde{p}Y\tilde{k}\tilde{p}Y\tilde{k})
=\frac{1}{4}(p-2\tilde{p}Y\tilde{k}+(\tilde{p}\tilde{p}Y)(\tilde{k}\tilde{k}Y))\\
&=\frac{1}{4}(p-2\tilde{p}Y\tilde{k}+p)
=\tilde{t}Y\tilde{t}=P \:.
\end{align*}

\vspace*{-1.5em}
\end{proof} 

\vspace*{0.2em}

\section{Derivation of the Perturbation Expansion for the Fermionic Projector} \label{sec4}
The definition of the fermionic projector \eqref{eq:defP} has the disadvantage that it
involves several hierarchies of infinite sums.
In this section we shall simplify the perturbation expansion of the fermionic projector
by bringing it into a form with the same structure as the perturbation expansion \eqref{k-pert} for $\tilde{k}$ involving only two hierarchies: $P$ should be an infinite sum of operator products of the form $b^<F_1bF_2b...bF_nb^>$ with $F_j\in\{p,k\}$, where the operators $b^<$, $b$, $b^>$
are defined by the perturbation series~\eqref{thm-defs}.
Again, the leading orders of the expansions derived in this section can be found in Appendix~\ref{appA}.

\begin{Thm}\label{maintheorem}
The fermionic projector $P_m$ defined in \eqref{eq:defP} can be written as
\begin{align}
	P_m=\frac{1}{2}(\tilde{p}^{\,\res}_m-\tilde{k}^{\,\res}_m)\:,
	\label{P-pq}
\end{align}
where the operators $\tilde{p}^{\,\res}_m$ and $\tilde{k}^{\,\res}_m$ are given by
\begin{align}
\tilde{p}^{\,\res}_m &=b^<_m\:\sum_{\beta=0}^{\infty}(i\pi)^{2\beta}p_m(b_mp_m)^{2\beta}\:b^>_m
	\label{ptil-thm} \\
\tilde{k}^{\,\res}_m &=b^<_m\:\sum_{r=0}^{\infty}\sum_{\rho=0}^{r}c(r,\rho)G_m(r,\rho)\:b^>_m \:.
\label{qtil-thm}
\end{align}
Here we have used the definitions
\begin{align}	c(r,\rho)=\pi^{2r}\: \frac{\Gamma(r-\rho+\frac{1}{2})}{\Gamma(-\rho+\frac{1}{2})\:r!}
\end{align}
and
\begin{align}	G_m(r,\rho)= \!\sum_{\substack{Q\in\mathcal{P}(2r+1)\\\#Q=2\rho}} \!\!\!\!\!\!
(-1)^{\sigma(r,\rho,Q)}\:F_m(Q,1)\, b_m\, F_m(Q,2) \,b_m.....b_m\, F_m(Q,2r+1)\:,
\end{align}
where $\mathcal{P}(2r+1)$ denotes the power set of $\{1,...,2r+1\}$, the
 operators $F(Q,n)$ are defined in \eqref{def-F} and
\beq 
\sigma(r,\rho,Q)=1+\sum_{ x\in\{1,...,2r+1\}\backslash Q}\:x\:. 
\eeq
\end{Thm} \noindent
Here the superscript~``$\res$'' indicates that the operators~$\tilde{p}^{\,\res}$ and~$\tilde{k}^{\,\res}$
have been rescaled using the procedure of Section~\ref{sec3}.
Alternatively, the superscript~``$\res$'' can be regarded as an abbreviation of the word ``residual,''
in view of the ``residual argument'' used in the light-cone expansion of the fermionic projector
(cf. Section~\ref{sec5}).
\begin{proof}[Proof of Theorem~\ref{maintheorem}]
We begin with formula \eqref{Psea-start},
\begin{align}
	P=\frac{1}{2}\:\breve{b}^<\:(p-\tilde{k}\tilde{p}Y)\:\breve{b}^> \:.\label{eq:P1}
\end{align}
Using that~$\tilde{p}=(p+X)^{1/2}$ and $Y=(p+X)^{-1}$, we obtain
\[ \tilde{p}Y=(p+X)^{-\frac{1}{2}}=p+\sum_{n\geq1}e_nX^n=\sum_{n\geq0}e_nX^n \]
with the combinatorial factors
\beq \label{endef}
e_n=(-1)^n\frac{(2n-1)!!}{2^nn!}\:,\qquad n\geq1\:.
\eeq
For the last identity, we have also set $e_0=1$ and $X^0=p$. Inserting into equation \eqref{eq:P1}, we obtain 
\begin{align}	P=\frac{1}{2}\,\breve{b}^<\:p\:\breve{b}^>-\frac{1}{2}\,\breve{b}^<\:\sum_{n\geq0}e_n\:\tilde{k}X^n\:\breve{b}^>
\label{eq:P2}
\end{align}
It is convenient to give the first series in \eqref{eq:P2} a name,
\begin{align}\label{p-res}	\tilde{p}^{\,\res}:=\breve{b}^<\:p\:\breve{b}^>=b^<\:\sum_{\beta=0}^{\infty}(i\pi)^{2\beta}p(bp)^{2\beta}\:b^>\:.
\end{align}
One can show in a way completely analogous to Proposition~\ref{thm:kpert} that 
\beq \label{pres}
	\tilde{p}^{\,\res}=\frac{1}{2\pi i}(\tilde{s}^+-\tilde{s}^-) \:,
\eeq
where $\tilde{s}^+$ and $\tilde{s}^-$ are the power series that arise formally analogous
to~\eqref{series-scaustilde},
\beq	\label{sv-series}
\tilde{s}_m^+=\sum_{n=0}^{\infty}(-s_m^+\mathscr{B})^ns_m^+ \:,\qquad
	\tilde{s}_m^-=\sum_{n=0}^{\infty}(-s_m^-\mathscr{B})^ns_m^- \:.
\eeq

The remaining task is to simplify the formula for the second series in \eqref{eq:P2},
which we now denote by~$\tilde{k}^{\,\res}$, i.e.
\begin{align}
\tilde{k}^{\,\res}:=\breve{b}^<\:\sum_{n\geq0}e_n\:\tilde{k}X^n\:\breve{b}^> \:,
	\label{def-qtil}
\end{align}
This will be done in Lemma \ref{lemma42}, concluding the proof of the theorem.
\QED

\begin{Remark}[\bf Stone's formula and polar decomposition] \em \label{remark}
In this remark we explain formal analogies of our perturbation expansion with
formulas known from functional analysis. First, for a selfadjoint
operator~$\mathcal{D}$ in a Hilbert space, one can compute the
spectral projector~$\mathscr{P}_m$ using Stone's formula (see e.g.~\cite[Theorem~VII.13]{reed+simon}),
\beq \label{specproj}
\mathscr{P}_m=\frac{1}{2\pi i}\left((\mathcal{D}-(m+i\varepsilon))^{-1}-(\mathcal{D}-(m-i\varepsilon))^{-1}\right).
\eeq
Formally applying this identity to the vacuum Dirac operator~$\mathcal{D}=i \slashed{\partial}$,
the resolvent is a multiplication operator in momentum space; namely,
\[ \Big( \mathcal{D}-(m \pm i\varepsilon) \Big)^{-1}(q) = \Big(\slashed{q}-(m\pm i \varepsilon) \Big)^{-1}
= \frac{\slashed{q}+m}{q^2-(m \pm i \varepsilon)^2} \:. \]
Comparing with~\eqref{p-ss} and \eqref{feynman}, we find that the spectral projector
coincides with the operator~$p_m$ with integral kernel \eqref{def-p-k},
\[\mathscr{P}_m= p_m = \frac{1}{2\pi i}\: (s_m^+-s_m^-)\:.\]

In the case of the interacting Dirac equation~$\mathcal{D}=i \slashed{\partial} + \B$, the resolvent
can be computed formally with a Neumann series,
\[ \Big( \mathcal{D}-(m \pm i\varepsilon) \Big)^{-1} = 
\sum_{n\geq0}(-s^{\pm}\mathscr{B})^ns^{\pm} \:. \]
This coincides precisely with the perturbation series~\eqref{sv-series}, and thus we can
write Stone's formula~\eqref{specproj} as
\[ \mathscr{P}_m = \frac{1}{2\pi i}(\tilde{s}^+_{m}-\tilde{s}^-_{m}) = \tilde{p}^{\,\res}_m\:, \]
with~$\tilde{p}^{\,\res}_m$ as given by~\eqref{pres}.
This consideration shows that the operator $\tilde{p}^{\,\res}$ defined in \eqref{p-res}
should be regarded as the spectral projector of the interacting Dirac equation. This is consistent with
the equations
\beq \label{rem-eq}
\tilde{p}^{\,\res}\tilde{k}^{\,\res}=\tilde{k}^{\,\res}=\tilde{k}^{\,\res}\tilde{p}^{\,\res}\qquad \text{and}\qquad \tilde{p}^{\,\res}\tilde{p}^{\,\res}=\tilde{p}^{\,\res}=\tilde{k}^{\,\res}\tilde{k}^{\,\res} \:,
\eeq
which can be computed explicitly from~\eqref{p-res} and~\eqref{def-qtil}.

Next, a bounded linear operator $K$ in a Hilbert space admits a polar decomposition of the form 
$$K=U|K|\,,$$
where~$|K|=\sqrt{K^*K}$ is the absolute value of~$K$ and~$U$ is an isometry
from the orthogonal complement of $\text{ker}(K)$ to the closure of the image of~$K$
(see e.g.~\cite[Theorem~VI.10]{reed+simon}).
From \eqref{eq:P1}, we see that
$$\tilde{k}^{\,\res}=\tilde{k}\tilde{p}Y=\tilde{k}\tilde{p}^{-1}=\tilde{k}|\tilde{k}|^{-1} \:. $$
We therefore regard $\tilde{k}^{\,\res}$ as the isometry transforming $|\tilde{k}|$ into $\tilde{k}$ and having the same range and kernel as $\tilde{k}$. Since $\tilde{p}$ is a positive operator, i.e. $\tilde{p}=|\tilde{p}|$, and $\tilde{p}^{\,\res}$ is the spectral projector and thus acts as the identity on the range of $\tilde{p}$, one can also regard~$\tilde{p}^{\,\res}$ as the isometry transforming~$|\tilde{p}|$ into~$\tilde{p}$ and having the same range and kernel as~$\tilde{p}$.
\QEDrem
\end{Remark}

It remains to provide the lemma quoted in the proof of Theorem \ref{maintheorem}.
\begin{Lemma}\label{lemma42}
The operator $\tilde{k}^{\,\res}$ defined by~\eqref{def-qtil} can be written as
\begin{align}\label{eq:qtilde}
\tilde{k}^{\,\res}=b^<\sum_{r=0}^{\infty}\sum_{\rho=0}^{r}c(r,\rho)\:
G(r,\rho)b^>
\end{align}
with the coefficients
\begin{align}
	c(r,\rho)=\pi^{2r}\frac{\Gamma(r-\rho+\frac{1}{2})}{\Gamma(-\rho+\frac{1}{2})r!}
\end{align}
and the operator products
\begin{align}	\label{Gdef}
G(r,\rho)=\sum_{\substack{Q\in\mathcal{P}(2r+1)\\\#Q=2\rho}}(-1)^{\sigma(r,\rho,Q)}\:F(Q,1) \:b\: F(Q,2) \:b \cdots b \:F(Q,2r+1)\:.
\end{align}
\end{Lemma}
\Proof
We introduce the abbreviation
\begin{align*}	U &:= p + X \\
&\overset{\eqref{X-1stversion}}{=} \sum_{\beta_1,\beta_2=0}^{\infty}(-i\pi)^{2(\beta_1+\beta_2)}(p+A)k(bk)^{2\beta_1}(p+A)^2k(bk)^{2\beta_2}(p+A)\:.
\end{align*}
Since $p$ acts as the identity on $U$, we can calculate the $n$-th power of $X$ to be
\begin{align*}
	X^n=(-p+U)^n=\sum_{l=0}^n {n\choose l} (-p)^{n-l}U^l=\sum_{l=0}^n {n\choose l} (-1)^{n-l}\:U^l \:,
\end{align*}
where the l-th power of $U$ is given by $U^0=p$ and
\[ U^l = (p+A) \left(\sum \nolimits_{\beta}(-i\pi)^{2\beta}k(bk)^{2\beta}(p+A)^2\right)^{2l-1}\sum
 \nolimits_{\alpha}(-i\pi)^{2\alpha}k(bk)^{2\alpha}(p+A) \]
if~$l >0$. Moreover, from \eqref{ktil-rewritten} we have
\begin{align*}
	\tilde{k}=\sum_{\beta=0}^{\infty}(-i\pi)^{2\beta}(p+A)k(bk)^{2\beta}(p+A) \:.
\end{align*} 
We thus obtain
\begin{align*}
	\tilde{k}X^n&=\sum_{l=0}^n {n\choose l} (-1)^{n-l}\:\tilde{k}\,U^l\\
	&=\sum_{l=0}^n {n\choose l} (-1)^{n-l}
	(p+A) \left(\sum \nolimits_{\beta}(-i\pi)^{2\beta}k(bk)^{2\beta}(p+A)^2\right)^{2l}\\
	&\hspace{5cm} \times\sum\nolimits_{\alpha}(-i\pi)^{2\alpha}k(bk)^{2\alpha}(p+A) \:.
\end{align*}
Inserting into \eqref{def-qtil} and using the calculation rules \eqref{2nd-simplerule}-\eqref{last-simplerule}, we obtain the following formula for~$\tilde{k}^{\,\res}$:
\begin{align}
	\nonumber \tilde{k}^{\,\res} =& \sum_{n=0}^{\infty}e_n\sum_{l=0}^n{n\choose l}(-1)^{n-l} \\
&\times b^<\:\left(\sum_{\beta=0}^\infty(-i\pi)^{2\beta}k(bk)^{2\beta}(p+\pi^2pbpbp)\right)^{2l}\sum_{\alpha=0}^\infty(-i\pi)^{2\alpha}k(bk)^{2\alpha}\:b^>\:.
\label{qtil-explicit}
\end{align}
Remarking that in \eqref{qtil-explicit} one gets one factor $\pi$ for each order of $b$,
for notational convenience we shall
omit the factors $\pi$ during the proof and reinsert them at the very end. Note moreover that for a given $n$, the $l$-sum in \eqref{qtil-explicit} equals $\tilde{k}X^n$ and is therefore of order~$\O(b^{2n})$ (as $X$ is of order $\O(b^{2})$, cf. \eqref{X-2ndversion}).
Thus for every $r\in\mathbb{N}$, the truncated sum 
\begin{align}\nonumber
	&\sum_{n=0}^{r}e_n\sum_{l=0}^n{n\choose l}(-1)^{n-l}\\		&\;\;\;\times b^<\:\left(\sum_{\beta}(-1)^{\beta}k(bk)^{2\beta}(p+pbpbp)\right)^{2l}\sum_{\alpha}(-1)^{\alpha}k(bk)^{2\alpha}\:b^>
	\label{trunc-sum}
\end{align}
coincides with $\tilde{k}^{\,\res}$ up to terms of order $\O(b^{2r+2})$. For this reason, we can calculate $\tilde{k}^{\,\res}$ from \eqref{trunc-sum} to every order in $b$. Interchanging the $n$- and $l$-sums
(which are both finite), we can carry out the $n$-sum, which involves only coefficients but no operator products, to obtain
\begin{align*}\nonumber
\tilde{k}^{\,\res} &=\sum_{l=0}^r \:\left[\sum_{n=l}^{r} e_n
	{n\choose l}(-1)^{n-l}\right] \:b^<\ldots b^> + \O(b^{2r+2}) \\
	&=\sum_{l=0}^r\, f_{l,r}\:b^<\:\left(\sum_{\beta}(-1)^{\beta}k(bk)^{2\beta}(p+pbpbp)\right)^{2l}\sum_{\alpha}(-1)^{\alpha}k(bk)^{2\alpha}\: b^>  + \O(b^{2r+2})
\end{align*}
with the coefficients
\[ f_{l,r}=\sum_{n=l}^{r} e_n
	{n\choose l}(-1)^{n-l}=\frac{(-1)^l}{l!}\frac{(2r+1)!!}{2^r(2l+1)(r-l)!}\:. \]
The combinatorics of the above operator products is analyzed in Lemma \ref{thm:combi-lemma}
below. Using the results of this Lemma, we obtain that for all $r\in\mathbb{N}$, the contributions to $\tilde{k}^{\,\res}$ involving $2r$ factors of $b$ are given by
\begin{align}\nonumber	b^<\:&\sum_{l=0}^r\sum_{\rho=0}^l\, f_{l,r}{{r+l-\rho}\choose{l-\rho}}\sum_{\substack{Q\in\mathcal{P}(2r+1)\\\#Q=2\rho}}(-1)^{1+\sum_{x\leq2r+1,x\notin Q}\,x} \\
& \hspace*{6cm} \times F(Q,1)bF(Q,2)b\cdots bF(Q,2r+1)\:b^> \nonumber \\
&=b^<\:\sum_{l=0}^r\sum_{\rho=0}^l \,f_{l,r}{{r+l-\rho}\choose{l-\rho}}G(r,\rho)\:b^>
\end{align}
(where we used the abbreviation~\eqref{Gdef}). Summing over all orders in~$b$, we
conclude that
\begin{align}	\tilde{k}^{\,\res}=b^<\:\sum_{r=0}^{\infty}\sum_{l=0}^r\sum_{\rho=0}^l\, f_{l,r}{{r+l-\rho}\choose{l-\rho}}G(r,\rho)\:b^> \:.
\end{align}
Since the operator products $G$ do not involve the index $l$, we may interchange the sums over $l$ and $\rho$ and perform the $l$-sum,
$$\sum_{l=\rho}^r\,f_{l,r}{{r+l-\rho}\choose{l-\rho}}=\frac{\Gamma(r-\rho+\frac{1}{2})}{\Gamma(\frac{1}{2}-\rho)\,r!} \:,$$
where we used the formula~\cite[eq.~(6.1.12)]{AS} to write the double factorial in terms
of the gamma function with half-integer argument.
Inserting back the factors $\pi$, we obtain the result.
\end{proof}

We finally prove the following combinatorial lemma.
\begin{Lemma}\label{thm:combi-lemma}
For any~$\rho, l, r \in \N_0$ with $\rho \leq r$, the following statements hold.
\begin{itemize}
\item[(i)] In the expression
\begin{align}	\Big[\sum_{\beta}(-1)^{\beta}k(bk)^{2\beta}(p+ pbpbp)\Big]^{2l}\sum_{\alpha}(-1)^{\alpha}k(bk)^{2\alpha} \:,
\label{combi-start}
\end{align}
only terms with an even number of factors $p$ appear.\\ 
\item[(ii)] In the series~\eqref{combi-start}, each term of the order~$b^{2r}$ of the form
\begin{align}
	C_1bC_2b...C_{2r}bC_{2r+1}\:,\qquad C_i\in\{p,k\}
	\label{op-prod}
\end{align}
appears exactly ${{r+l-\rho}\choose{l-\rho}}$ times, where~$2\rho$ denotes the number of factors~$p$
in the product. Here we adopt the convention ${n\choose n'}=0$ if $n'<0$.\\
Moreover, each term of the form~\eqref{op-prod} appearing in~\eqref{combi-start} has the sign
 \[  (-1)^{\sum_{x \in J} x} \qquad \text{where} \qquad
 J= \left\{i\in\{1,...,2r+1\}\text{ with } C_i=k \right\} . \]
\end{itemize}
\end{Lemma}
\Proof
(i) We proceed by induction in $l$. For $l=0$, the number of factors $p$ in \eqref{combi-start} is
obviously zero. Thus assume that the induction hypothesis holds for a given l. Then
\begin{align*}	&\Big[\sum_{\beta}(-1)^{\beta}k(bk)^{2\beta}(p+ pbpbp)\Big]^{2(l+1)}\sum_{\alpha}(-1)^{\alpha}k(bk)^{2\alpha}\\
&=\Big[\sum_{\beta}(-1)^{\beta}k(bk)^{2\beta}(p+pbpbp)\Big]^{2}\times\\
&\;\;\;\;\times\Big[\sum_{\beta}(-1)^{\beta}k(bk)^{2\beta}(p+ pbpbp)\Big]^{2l}\sum_{\alpha}(-1)^{\alpha}k(bk)^{2\alpha}\\
&=\Big[\sum_{\beta}(-1)^{\beta}k(bk)^{2\beta}(p+ pbpbp)\Big]^{2}\times(\mbox{terms with $\#p$ even})\\
&=\sum_{\beta_1,\beta_2}(-1)^{\beta_1+\beta_2}\Big[(kb)^{2\beta_1}p(bk)^{2\beta_2}+(kb)^{2\beta_1}p(bk)^{2\beta_2}bpbp+(kb)^{2\beta_1}kbpbk(bk)^{2\beta_2}+\\
&\;\;\;\;+(kb)^{2\beta_1}kbpbk(bk)^{2\beta_2}bpbp\Big]\times(\mbox{terms with $\#p$ even}).
\end{align*}
Now it can be easily checked that when a term with an even number of factors $p$ is multiplied from the left by one of the terms in the last sum, then the number of factors~$p$ either does not change or is raised by two. So in any case one obtains a term with an even number of factors $p$.\\[0.5em]
(ii) We again proceed by induction in $l$. For $l=0$, there are obviously no contributions
for~$\rho>0$. In the case~$\rho=0$, the term of order $b^{2r}$ in \eqref{combi-start} is $(-1)^rk(bk)^{2r}$. It appears exactly 
$$1={r\choose0}={{r+l-\rho}\choose{l-\rho}}$$ 
times in \eqref{combi-start} and has the sign
$$(-1)^r=(-1)^{3r+1}=(-1)^{4r^2+3r+1}=(-1)^{\sum_{i=1}^{2r+1}i}=(-1)^{\sum_{i\in J}i}.$$ 
This proves the claim in the case~$l=0$.

Now assume that the induction hypothesis holds for a given~$l$. Consider an operator product
of the form~\eqref{op-prod} which involves~$2\rho$ factors of~$p$. Our task is to
determine how often and with which sign this operator product appears in the series
\begin{align}	\Big[\sum_{\beta}(-1)^{\beta}k(bk)^{2\beta}(p+pbpbp)\Big]^{2(l+1)}\sum_{\alpha}(-1)^{\alpha}k(bk)^{2\alpha} \:.
\label{term-l+1}
\end{align}
Let us first count how often it appears.
Since the operator product \eqref{op-prod} is of the order $b^{2r}$, we only need to calculate up to
order $\O(b^{2r+2})$,
\begin{align}\nonumber	\eqref{term-l+1}&=\Big[\sum_{\beta=0}^r(-1)^{\beta}k(bk)^{2\beta}(p+pbpbp)\Big]^{2(l+1)}\sum_{\alpha=0}^r(-1)^{\alpha}k(bk)^{2\alpha}+\O(b^{2r+2})\\\nonumber
=&\Big[\sum_{\beta=0}^r\ldots \Big]^2\times\Big[\sum_{\beta=0}^r\ldots \Big]^{2l}\sum_{\alpha=0}^r\ldots +\O(b^{2r+2})\\\nonumber
=&\Big[\underbrace{\sum_{\beta=0}^r(-1)^{\beta}(kb)^{2\beta}k}_{=:A}+\underbrace{\sum_{\beta=0}^{r-1}(-1)^{\beta}(kb)^{2\beta}kbpbp}_{=:B}\Big]^2\times
\underbrace{ \Big[\sum_{\beta=0}^r\ldots \Big]^{2l}\sum_{\alpha=0}^r \ldots}_{=:L}
 +\O(b^{2r+2}) \nonumber \\ 
=&\Big[AA+AB+BA+BB \Big]\:L +\O(b^{2r+2})\:, \label{2term-l+1}
\end{align}
where
\begin{align*}
AA &= \sum_{\beta_1=0}^r\sum_{\beta_2=0}^r(-1)^{\beta_1+\beta_2}(kb)^{2\beta_1}p(bk)^{2\beta_2}\:, \\
AB &=\sum_{\beta_1=0}^r\sum_{\beta_2=0}^{r-1}(-1)^{\beta_1+\beta_2}(kb)^{2\beta_1}p(bk)^{2\beta_2}bpbp\:, \\
BA &= \sum_{\beta_1=0}^{r-1}\sum_{\beta_2=0}^r(-1)^{\beta_1+\beta_2}(kb)^{2\beta_1}kbpbk(bk)^{2\beta_2}\:, \\
BB &= \sum_{\beta_1=0}^{r-1}\sum_{\beta_2=0}^{r-1}(-1)^{\beta_1+\beta_2}(kb)^{2\beta_1}kbpbk(bk)^{2\beta_2}bpbp\:, \\
L &= \Big[ \sum_{\beta=0}^r (-1)^{\beta}k(bk)^{2\beta}(p+pbpbp) \Big]^{2l}\sum_{\alpha=0}^r
(-1)^{\alpha}k(bk)^{2\alpha} \:.
\end{align*}

We now treat the three cases~$\rho=l+1\leq r$, $\rho=0$, and~$0<\rho<l+1$ separately.
\begin{itemize}
\item[(a)] $\rho=l+1\leq r$: \\
By the induction hypothesis, the operator product~\eqref{op-prod} appears in~$L$ exactly ${{r+l-\rho}\choose{l-\rho}}={{r-1}\choose{-1}}=0$ times. Thus the term~\eqref{op-prod}
arises in~\eqref{2term-l+1} if a
summand of $L$ which involves~$2\rho-2$ factors of~$p$ is multiplied from the left by a term of the series~$[\sum_{\beta=0}^r\ldots ]^2$ in such a way that two additional factors $p$ are created. This can happen in the following ways: \\
\begin{tabular}{ll}
(1) & ($AA$-term ending by $k$)$\times$(term in~$L$ beginning by $k$)\\
(2) & ($AB$-term)$\times$(term in~$L$ beginning arbitrarily)\\
(3) & ($BA$-term)$\times$(term in~$L$ beginning by $k$)\\
(4) & ($BB$-term)$\times$(term in~$L$ beginning arbitrarily).\\
\end{tabular}

\noindent Each term of the form~\eqref{op-prod} arises as such a product precisely once, because
\begin{tabular}{l}
in case (1) the first $p$ is at an \textit{odd} and the second $p$ is at an \textit{odd} position,\\
in case (2) the first $p$ is at an \textit{odd} and the second $p$ is at an \textit{even} position,\\
in case (3) the first $p$ is at an \textit{even} and the second $p$ is at an \textit{odd} position,\\
in case (4) the first $p$ is at an \textit{even} and the second $p$ is at an \textit{even} position.\\
\end{tabular}

\noindent Since by the induction hypothesis, a term involving~$2\rho-2$ factors of~$p$ appears in~$L$ exactly 
${{r}\choose{0}}=1$ times, the term \eqref{op-prod} appears in \eqref{2term-l+1} exactly
 $$1={{r}\choose{0}}={{r+(l+1)-\rho}\choose{(l+1)-\rho}}\;\;\;\mbox{times}\:.$$ 

\item[(b)] $\rho=0$:\\
In this case, the operator product~\eqref{op-prod} is of the form~$k(bk)^{2r}$.
Only an $AA$-term in~\eqref{2term-l+1} with $\beta_2=0$
does not create an additional factor $p$ when it is multiplied from the right by a
term of~$L$ involving no factors of~$p$.
By induction hypothesis, for every~$r'\in\mathbb{N}$ the factor $k(bk)^{2r'}$ appears in~$L$
exactly~${{r'+l}\choose{l}}$ times. Thus the factor $k(bk)^{2r}$ appears in \eqref{2term-l+1} exactly
$$\sum_{r'=0}^r{{r'+l}\choose{l}}={{r+(l+1)-0}\choose{(l+1)-0}}={{r+(l+1)-\rho}\choose{(l+1)-\rho}}\;\;\;\mbox{times}\:. $$

\item[(c)] $0<\rho<l+1$: \\
As in case~(a), we want to factor a given operator product~\eqref{op-prod}
at a given odd position into a product of a term in $[\sum_{\beta=0}^r\ldots ]^2$ and a term in~$L$.
Exactly as in case~(a) one verifies that if this factorization is possible, it is unique.
Furthermore, this factorization is possible at the first position, at the third position, \ldots,
until the second factor~$p$ appears (more precisely, until we are at the position of the second factor~$p$,
or else the second factor $p$ is at an even position and we are at the subsequent odd position).
We denote the number of such possible factoring positions by~$\lambda+1$.
For $1\leq r'\leq\lambda$, factoring at the $r'$-th position yields a term in~$L$ of order $b^{2(r+1-r')}$
involving $2\rho$ factors of~$p$. According to the induction hypothesis, this term in~$L$ appears with
the combinatorial factor~${{(r+1-r')+l-\rho}\choose{l-\rho}}$.
Factoring at the last position~$r'=\lambda+1$ yields a term in~$L$ of order $b^{2(r-\lambda)}$
involving $2\rho-2$ factors of $p$. According to the induction hypothesis, this term in~$L$ appears with
the combinatorial factor~${{(r-\lambda)+l-(\rho-1)}\choose{l-(\rho-1)}}$.
Therefore, the operator product~\eqref{op-prod} appears in the series~\eqref{2term-l+1} exactly
\begin{equation*}
\qquad \sum_{r'=1}^{\lambda}{{(r+1-r')+l-\rho}\choose{l-\rho}}+{{(r-\lambda)+l-(\rho-1)}\choose{l-(\rho-1)}}
= {{r+(l+1)-\rho}\choose{(l+1)-\rho}}
\end{equation*}
times.
\end{itemize}
We have thus proved the first part of Lemma \ref{thm:combi-lemma}~(ii).

It remains to prove the induction step for the claim concerning the signs of the terms \eqref{op-prod} in \eqref{2term-l+1}. By induction hypothesis, an operator product~$\mathcal{H}$ 
of the form \eqref{op-prod} appears in~$L$ with the sign
$$(-1)^{1+\Sigma\;i}\:, $$
where the sum $\Sigma$ goes over all positions in $\mathcal{H}$ where $C_i=k$. This is of course equal to
$$(-1)^{1+\#(\text{factors~$k$ at odd positions in }\mathcal{H})} \:. $$
Likewise, each operator product $\mathcal{G}$ appearing in $(A+B)^2$ has the sign
$$(-1)^{\#(\text{factors~$k$ at odd positions in }\mathcal{G})} \:. $$
When multiplying~$\mathcal{G}$ by~$\mathcal{H}$, a position in the resulting operator
product~$\mathcal{G}\mathcal{H}$ is odd iff the corresponding position in the old
product $\mathcal{G}$ or $\mathcal{H}$ was odd. We now distinguish two cases.
\begin{itemize}
\item[(a)] The term $\mathcal{G}$ ends with a factor $k$ and the term $\mathcal{H}$ begins with a factor $k$: \\
As these two factors $k$ multiply to a factor $p$ by our multiplication rules, the resulting term $\mathcal{G}\mathcal{H}$ has two factors $k$ less at odd positions than the initial factors 
$\mathcal{G}$ and $\mathcal{H}$ together. Thus the sign of $\mathcal{G}\mathcal{H}$ is
\begin{align*}
	&(-1)^{\#(\text{factors~$k$ at odd positions in $\mathcal{G}$})}\:
	(-1)^{1+\#(\text{factors~$k$ at odd positions in $\mathcal{H}$})}\\
	&=(-1)^{\#(\text{factors~$k$ at odd positions in $\mathcal{G}$})+1+\#(\text{factors~$k$ at odd positions in $\mathcal{H}$})-2}\\
	&=(-1)^{1+\#(\text{factors~$k$ at odd positions in $\mathcal{G} \mathcal{H}$})}
\end{align*}
\item[(b)] In all other cases, one has
\begin{align*}
	\#(k\mbox{ at odd positions in }\mathcal{G}\mathcal{H})=&\:\#(k\mbox{ at odd positions in }\mathcal{G})\\&+\#(k\mbox{ at odd positions in }\mathcal{H}),
\end{align*}
and therefore $\mathcal{G}\mathcal{H}$ has the sign
\begin{align*}
	&(-1)^{\#(\text{factors~$k$ at odd positions in $\mathcal{G}$})}(-1)^{1+\#(\text{factors~$k$ at odd positions in $\mathcal{H}$})}\\
	&=(-1)^{1+\#(\text{factors~$k$ at odd positions in $\mathcal{G} \mathcal{H}$})} \:.
\end{align*}
\end{itemize}
This concludes the proof.
\end{proof}

\section{The Unitary Perturbation Flow} \label{sec5}
In this section we prove that there exists a operator $U$, uniquely defined by a perturbation series, which transforms the vacuum operators $p_m$, $k_m$ into the interacting operators
$\tilde{p}^{\,\res}_m$, $\tilde{k}^{\,\res}_m$, i.e.
\beq \label{result-unitary}
\tilde{p}^{\,\res}_m=Up_mU^{-1}\;,\qquad \tilde{k}^{\,\res}_m=Uk_mU^{-1}\:.
\eeq
The operator~$U$ will be unitary with respect to the indefinite inner product on the wave functions
\beq \label{iprod}
\int \overline{\Psi(x)} \Phi(x) \: d^4x \:,
\eeq
where~$\overline{\Psi} = \Psi^\dagger \gamma^0$ is the usual adjoint spinor.
For the proof, we shall consider the one-parameter family of rescaled operators
\beq \label{not-ptau}
\tilde{p}_m^\tau:=(\tilde{p}_m^{\,\res})^\tau \qquad \text{and}
\qquad \tilde{k}_m^\tau:=(\tilde{k}_m^{\,\res})^\tau
\eeq
corresponding to the family of external fields $(\tau\mathscr{B})_{\tau\geq0}$, i.e.
\beq \label{defeq-ptau}
(i\slashed{\partial}_x + \tau{\mathscr{B}}(x) - m)\:\tilde{p}_m^\tau(x,y)=0=(i\slashed{\partial}_x + \tau{\mathscr{B}}(x) - m)\:\tilde{k}_m^\tau(x,y) \:.
\eeq
The idea is to consider the differential equations for
\[ \frac{d}{d\tau} \tilde{p}_m^\tau \qquad \text{and} \qquad
\frac{d}{d\tau} \tilde{k}_m^\tau\:, \]
the so-called {\em{perturbation flow equations}},
and to write them purely in terms of interacting operators. This will reveal that
the perturbation flow equations have the commutator structure
\beq\label{pk-deriv-commutator}
\frac{d}{d\tau}\tilde{p}_m^\tau=\Big[A(\tau),\tilde{p}_m^\tau\Big]\qquad\text{and}\qquad \frac{d}{d\tau}\tilde{k}_m^\tau=\Big[A(\tau),\tilde{k}_m^\tau\Big]
\eeq
with an anti-symmetric operator~$A$, from which~$U$ can be obtained by integration.

In preparation, we define the operator 
\beq\label{def-stilde}
\tilde{s}_{m}:=\frac{1}{2}(\tilde{s}_m^++\tilde{s}_m^-)\:,
\eeq
noting that this implies the relations 
\beq \label{rel-stil-ptil}
\tilde{s}_m^\pm=\tilde{s}_m\pm i\pi \tilde{p}_m\:.
\eeq
Moreover, we can generalize the formulas \eqref{eq:ps-p} and \eqref{eq:ks-k}.
\begin{Lemma}\label{lemma51}
The operators $\tilde{p}^{\,\res}_m,\, \tilde{k}^{\,\res}_m$ and $\tilde{s}_m$ satisfy
the following identities:
\begin{align*}
	\tilde{p}^{\,\res}_m\, \tilde{s}_{m'} &=\tilde{s}_{m'}\, \tilde{p}^{\,\res}_m = \frac{\pv}{m-m'}
\:\tilde{p}^{\,\res}_m \\
	\tilde{k}^{\,\res}_m\, \tilde{s}_{m'} &=\tilde{s}_{m'}\, \tilde{k}^{\,\res}_m = \frac{\pv}{m-m'}\:
\tilde{k}^{\,\res}_m \:.
\end{align*}
\end{Lemma}
\Proof
In view of Theorem \ref{maintheorem}, it suffices to show that 
$$p_mb^>_m\,\tilde{s}_{m'}=\pvm \:p_mb^>_m\qquad\text{and}\qquad \tilde{s}_{m'}\,b^<_mp_m=\pvm \:b^<_mp_m\:.$$
Substituting the relations
$$s_m^{\pm}=s_m\pm i\pi p_m$$
into the power series \eqref{sv-series} of $\tilde{s}_m^{\pm}$, we note that the terms with an odd number of factors $p_m$ cancel in \eqref{def-stilde}. We thus obtain
\begin{align*}
\tilde{s}_{m}=&\sum_{n\geq0}(-1)^n\sum_{\substack{Q\in \mathcal{P}(n+1)\\\# Q\text{ even}}}(i\pi)^{\# Q}\,C_m(Q,1)\mathscr{B}C_m(Q,2)\ldots \mathscr{B}C_m(Q,n+1)\\
=&\sum_{n\geq0}(-1)^n\sum_{\substack{Q\in \mathcal{P}(n)\\\# Q\text{ even}}}(i\pi)^{\# Q}\,s_m\mathscr{B}C_m(Q,2)\ldots \mathscr{B}C_m(Q,n+1)\\
&+\sum_{n\geq0}(-1)^n\sum_{\substack{Q\in \mathcal{P}(n)\\\# Q\text{ odd}}}(i\pi)^{\# Q}\,p_m\mathscr{B}C_m(Q,2)\ldots \mathscr{B}C_m(Q,n+1) \:,
\end{align*}
where $\mathcal{P}(n+1)$ is defined as in Theorem \ref{maintheorem} and 
$$C_m(Q,j):=\begin{cases}p_m\:,\quad j\in Q\\s_m\:,\quad j\notin Q \:.\end{cases}$$
Moreover, we introduce the abbreviations
$$\sum_\text{n,even}:=\sum_{n\geq0}(-1)^n\sum_{\substack{Q\in \mathcal{P}(n)\\\# Q\text{ even}}}(i\pi)^{\# Q}\qquad\text{and}\qquad\sum_\text{n,odd}:=\sum_{n\geq0}(-1)^n\sum_{\substack{Q\in \mathcal{P}(n)\\\# Q\text{ odd}}}(i\pi)^{\# Q} \:.$$
Applying the calculation rules from Lemma \ref{lemma21}, we obtain
\begin{align*}
	p_m&b^>_m \,\tilde{s}_{m'} = p_m\sum_{k\geq0}(-\mathscr{B}s_m)^{k}\:\tilde{s}_{m'}\\
	=&p_m \sum_\text{n,even}s_{m'}\mathscr{B}C_{m'}(Q,2)\ldots \mathscr{B}C_{m'}(Q,n+1)\\
	&-p_m\sum_{k\geq1}(-\mathscr{B}s_m)^{k-1}\mathscr{B}s_m \sum_\text{n,even}s_{m'}\mathscr{B}C_{m'}(Q,2)\ldots \mathscr{B}C_{m'}(Q,n+1)\\
	&+p_m \sum_\text{n,odd}p_{m'}\mathscr{B}C_{m'}(Q,2)\ldots \mathscr{B}C_{m'}(Q,n+1)\\
	&-p_m\sum_{k\geq1}(-\mathscr{B}s_m)^{k-1}\mathscr{B}s_m \sum_\text{n,odd}p_{m'}\mathscr{B}C_{m'}(Q,2)\ldots \mathscr{B}C_{m'}(Q,n+1)\\
	=&\pvm \sum_\text{n,even}p_{m}\mathscr{B}C_{m'}(Q,2)\ldots \mathscr{B}C_{m'}(Q,n+1)\\
	&-\pvm \:p_m \sum_{k\geq0}(-\mathscr{B}s_m)^{k}\mathscr{B}s_m\sum_\text{n,even}\mathscr{B}C_{m'}(Q,2)\ldots \mathscr{B}C_{m'}(Q,n+1)\\
	&+\pvm \:p_m \sum_{k\geq0}(-\mathscr{B}s_m)^{k}\mathscr{B}s_{m'}\sum_\text{n,even}\mathscr{B}C_{m'}(Q,2)\ldots \mathscr{B}C_{m'}(Q,n+1)\\	
	&-\pi^2\delta(m-m')\:p_m\sum_{k\geq0}(-\mathscr{B}s_m)^{k}\mathscr{B}\sum_\text{n,even}p_m\mathscr{B}C_{m}(Q,2)\ldots \mathscr{B}C_{m}(Q,n+1)\\
	&+\delta(m-m')\sum_\text{n,odd}p_{m}\mathscr{B}C_{m}(Q,2)\ldots \mathscr{B}C_{m}(Q,n+1)\\
	&-\frac{\pv}{m'-m} \:p_m\sum_{k\geq0}(-\mathscr{B}s_m)^{k}\mathscr{B}\sum_\text{n,odd}p_{m'}\mathscr{B}C_{m'}(Q,2)\ldots \mathscr{B}C_{m'}(Q,n+1) \:.
\end{align*}
Now a careful inspection shows that in the last equation, the fourth and the fifth line cancel each other,
whereas the $k=0$-terms in the third and the sixth line cancel the $n\geq1$-terms in the first line.
Furthermore, the $k\geq1$-terms in the third and sixth line cancel the $n\geq1$-terms in the second line. We thus conclude
$$p_mb^>_m\,\tilde{s}_{m'}=\pvm (p_m+p_m\sum_{k\geq1}(-\mathscr{B}s_m)^{k})=\pvm
p_mb^>_m \:. $$
The equation $\tilde{s}_{m'}\,b^<_mp_m=b^<_mp_m$ can be proven analogously.
\QED
Next we need to derive a completeness result for the spectral projectors of the interacting
Dirac operator. The completeness relation could be stated as
\[ \int_{\R \cup i \R} \tilde{p}_m^\text{res}\:dm = \1\:. \]
However, since we introduced the spectral projectors only for a real and positive
mass parameter~$m$, it is more convenient to work instead with contour integrals of the
Green's function in the complex plane. To this end, we introduce for any~$\varepsilon>0$
the hyperbolas
\[ C_\varepsilon = \left\{ \mu \in \C \text{ with } |(\text{Re} \,\mu)(\text{Im} \,\mu)|
= \varepsilon^2  \right\} . \]
We also consider~$C_\varepsilon$ as a contour along which we integrate in anti-clockwise orientation.
In view of the identity~$(p\slsh)^2=p^2 \in \R$, the eigenvalues of the matrix~$p \slsh$
lie in $\R \cup i \R$. Hence for any~$\mu \in C_\varepsilon$, the free Green's function $s_\mu$ can be defined in momentum space by
\beq \label{smudef}
s_\mu(p) = (p\slsh - \mu)^{-1}\:.
\eeq
Moreover, the Green's function with interaction can be defined by the perturbation series
\beq \label{sper}
\tilde{s}_\mu = \sum_{n=0}^\infty (-s_\mu \B)^n s_\mu \:,
\eeq
because writing the Feyman diagrams in momentum space, one gets products of the form
\beq \label{feynmanmom}
s_\mu(p_n) \:\B(p_n-p_{n-1}) \cdots s_\mu(p_1) \:\B(p_1-p_0)\: s_\mu(p_0)
\eeq
where each factor is well-defined according to~\eqref{smudef}.
Next, we define the principal value by
\beq \label{princ}
\int_{C^\varepsilon} \frac{\text{PP}}{\mu - x}\:\cdots\: \tilde{s}_\mu\: d\mu 
= \int_{C^\varepsilon}  \frac{1}{2} \left( \frac{1}{\mu - x + 4\varepsilon e^{i \varphi}}
+ \frac{1}{\mu - x - 4\varepsilon e^{i \varphi}} \right)\:\cdots\: \tilde{s}_\mu\: d\mu \:,
\eeq
where the phase $\varphi=\varphi(x)$ is to be chosen such that the points~$x \pm 4 \varepsilon
e^{i \varphi}$ are both outside the set
\[ J_\varepsilon := \left\{ \mu \in \C \text{ with } |(\text{Re} \,\mu)
(\text{Im} \,\mu)| \leq \varepsilon^2 \right\} \]
(in the limit~$\varepsilon \searrow 0$,
the results of all contour integrals will become independent of the choice of~$\varphi$).
\begin{Lemma} \label{lemma52} For any~$\mu \in C_\varepsilon$ and~$m>0$, the following identities
hold to every order in perturbation theory, with convergence in the distributional sense:
\begin{align}
-\frac{1}{2 \pi i} \: \lim_{\varepsilon \searrow 0} \int_{C_\varepsilon} \tilde{s}_\mu\:d\mu &= \1
\label{ci1} \\
-\frac{1}{2 \pi i} \: \lim_{\varepsilon \searrow 0} \int_{C_\varepsilon} 
\frac{\text{\rm{PP}}}{\mu-m}\: \tilde{s}_\mu \:d\mu &= \tilde{s}_m \label{ci2} \\
\tilde{s}_\mu\, \tilde{p}^{\,\res}_m = \frac{1}{m-\mu}
\:\tilde{p}^{\,\res}_m \:,\qquad&
\tilde{s}_\mu\, \tilde{k}^{\,\res}_m = \frac{1}{m-\mu}\:
\tilde{k}^{\,\res}_m \label{ci0}
\end{align}
Here we used the definitions~\eqref{smudef}, \eqref{sper} and~\eqref{princ},
whereas~$\tilde{s}_m$ is defined by~\eqref{def-stilde}.
\end{Lemma}
\Proof In order to derive~\eqref{ci1}, we consider the perturbation series in momentum space.
Using again that the spectrum of the matrix~$p\slsh$ lies in~$\R \cup i\R$, we
can integrate~\eqref{smudef} with residues to obtain
\[ \int_{C_\varepsilon} \frac{1}{p\slsh- \mu} = -2 \pi i \: \1\:. \]
Similarly, when integrating~\eqref{feynmanmom}, the integrand is a product of
poles $(\mu-x)^{-p}$ of total order larger than one, and the
poles are all enclosed by the contour~$C_\varepsilon$. Hence computing the
integral with residues, we obtain zero. This proves~\eqref{ci1}.

For any~$\mu \in C_\varepsilon$ and~$\nu \in \C \setminus J_\varepsilon$,
a straightforward calculation using~\eqref{smudef} and~\eqref{sper} shows that
to every oder in perturbation theory, the ``resolvent identity''
\beq \label{ri}
\tilde{s}_\mu \tilde{s}_\nu = \frac{1}{\mu - \nu} \left( \tilde{s}_\mu - \tilde{s}_\nu \right)
\eeq
holds. Integrating~$\mu$ on both sides over~$C_\varepsilon$ and using that
\begin{align*}
\int_{C_\varepsilon} \tilde{s}_\mu \tilde{s}_\nu \: d\mu &= 
\left( \int_{C_\varepsilon} \tilde{s}_\mu \: d\mu \right) \tilde{s}_\nu 
\overset{\eqref{ci1}}{=} -2 \pi i \,\tilde{s}_\nu \\
\int_{C_\varepsilon} \frac{1}{\mu - \nu} \:\tilde{s}_\nu \: d\mu &=
\tilde{s}_\nu \int_{C_\varepsilon} \frac{1}{\mu - \nu} \: d\mu = 0 \:,
\end{align*}
we find that
\[ \int_{C_\varepsilon} 
\frac{1}{\mu-\nu}\: \tilde{s}_\mu \:d\mu = -2 \pi i\, \tilde{s}_\nu \:. \]
Now~\eqref{ci2} follows immediately from the definition of the principal value~\eqref{princ}.

Similar as in the proof of Lemma~\ref{lemma51}, the identities~\eqref{ci0} follow
from the relations
\[ s_\mu b^<_m p_m = \frac{1}{m-\mu}\: b^<_m p_m \qquad \text{and} \qquad
s_\mu b^<_m k_m = \frac{1}{m-\mu}\: b^<_m k_m\:. \]
These relations follow immediately from the definition of~$b^<_m$, \eqref{thm-defs},
by applying the resolvent identity~\eqref{ri} (with~$\nu$ replaced by~$m$)
as well as the relations~\eqref{eq:ps-p} and~\eqref{eq:ks-k}, which after replacing~$m'$
by the complex parameter~$\mu \in C_\varepsilon$ are clearly valid even without the principal part.
\QED

We are now ready to state the equations for the perturbation flow. Using a notation similar
to~\eqref{not-ptau} and~\eqref{defeq-ptau} for the Greens functions corresponding
to the family of external fields $(\tau\mathscr{B})_{\tau\geq0}$, we write
$$\tilde{s}_m^\tau\;,\qquad (\tilde{s}_m^+)^\tau \qquad \text{ and}\qquad (\tilde{s}_m^-)^\tau\:. $$

\begin{Thm} \label{pfeq}
The perturbation flow equations have the commutator structure
$$\frac{d}{d\tau}\tilde{p}_m^\tau=\Big[A(\tau),\tilde{p}_m^\tau\Big]\qquad\text{and}\qquad \frac{d}{d\tau}\tilde{k}_m^\tau=\Big[A(\tau),\tilde{k}_m^\tau\Big],$$
where $A(\tau)$ is a family of operators, which are anti-symmetric with respect to the
inner product~\eqref{iprod}.
\end{Thm}
\Proof
Since the Greens functions $(\tilde{s}_m^\pm)^\tau$ are defined via the formal Neumann series \eqref{sv-series}, 
their $\tau$-derivatives are given by
$$\frac{d}{d\tau}(\tilde{s}_m^\pm)^\tau=-(\tilde{s}_m^\pm)^\tau\mathscr{B}(\tilde{s}_m^\pm)^\tau \:. $$
Using the definition of $\tilde{p}_m^\tau$, \eqref{pres}, and applying the relations \eqref{rel-stil-ptil}, we obtain 
\beq \label{struc-deriv-ptil}
\frac{d}{d\tau}\tilde{p}_m^\tau=-\tilde{s}_m^\tau\mathscr{B}\tilde{p}_m^\tau-\tilde{p}_m^\tau\mathscr{B}\tilde{s}_m^\tau \:.
\eeq
The identity
$$0=\frac{d}{d\tau}\Big(i\slashed{\partial}+\tau\B -m\Big)\tilde{k}_m^\tau=\B\tilde{k}_m^\tau+\Big(i\slashed{\partial}+\tau\B -m\Big)\frac{d}{d\tau}\tilde{k}_m^\tau$$
implies that the $\tau$-derivative of $\tilde{k}_m^\tau$ is a solution of the inhomogeneous equation
$$\Big(i\slashed{\partial}+\tau\B -m\Big)\frac{d}{d\tau}\tilde{k}_m^\tau=-\B\tilde{k}_m^\tau\:.$$
Thus, the $\tau$-derivative of $\tilde{k}_m^\tau$ is of the form
$$
\frac{d}{d\tau}\tilde{k}_m^\tau=-\tilde{s}_m^\tau\mathscr{B}\tilde{k}_m^\tau+\mathscr{H}_m^\tau
\:,
$$
where $\mathscr{H}_m^\tau$ is a solution of the homogeneous equation
\beq\Big(i\slashed{\partial}+\tau\B -m\Big)\mathscr{H}_m^\tau=0 \:.\label{eq:unit-homog} \eeq
Since $\tilde{k}_m^\tau$ is a symmetric operator, the same is true for its $\tau$-derivative. We thus conclude that
$$\mathscr{H}_m^\tau=-\tilde{k}_m^\tau\mathscr{B}\tilde{s}_m^\tau+\hat{\mathscr{H}}_m^\tau\,,$$
where $\hat{\mathscr{H}}_m^\tau$ is another solution of the homogeneous equation \eqref{eq:unit-homog}, which is moreover symmetric with respect to the inner product~\eqref{iprod}.
The most general ansatz for $\hat{\mathscr{H}}_m^\tau$ is
$$\hat{\mathscr{H}}_m^\tau=\tilde{p}_m^\tau\mathcal{A}_m^\tau\tilde{p}_m^\tau+\tilde{k}_m^\tau\mathcal{C}_m^\tau\tilde{k}_m^\tau+\tilde{k}_m^\tau\mathcal{D}_m^\tau\tilde{p}_m^\tau+\tilde{p}_m^\tau(\mathcal{D}_m^\tau)^*\tilde{k}_m^\tau$$
with a certain operator $\mathcal{D}_m^\tau$ and certain symmetric operators $\mathcal{A}_m^\tau$ and $\mathcal{C}_m^\tau$. This leads to the identity
\begin{align*}	\frac{d}{d\tau}(\tilde{k}_m^\tau\tilde{k}_{m'}^\tau) =-\tilde{k}&_m^\tau \mathscr{B}\tilde{s}_m^\tau\tilde{k}_{m'}^\tau-\tilde{k}_m^\tau\tilde{s}_{m'}^\tau\mathscr{B}\tilde{k}_{m'}^\tau\\
+\delta(m-m')&\Big(-\tilde{s}_m^\tau\mathscr{B}\tilde{p}_m^\tau-\tilde{p}_m^\tau\mathscr{B}\tilde{s}_m^\tau
+\tilde{p}_m^\tau\mathcal{A}_m^\tau\tilde{k}_m^\tau+\tilde{k}_m^\tau\mathcal{C}_m^\tau\tilde{p}_m^\tau
+\tilde{k}_m^\tau\mathcal{D}_m^\tau\tilde{k}_m^\tau\\
&\;\;\:+\tilde{p}_m^\tau(\mathcal{D}_m^\tau)^*\tilde{p}_m^\tau
+\tilde{k}_m^\tau\mathcal{A}_m^\tau\tilde{p}_m^\tau+\tilde{p}_m^\tau\mathcal{C}_m^\tau\tilde{k}_m^\tau+\tilde{p}_m^\tau\mathcal{D}_m^\tau\tilde{p}_m^\tau+\tilde{k}_m^\tau(\mathcal{D}_m^\tau)^*\tilde{k}\Big).
\label{eq:}
\end{align*}
In view of Lemma \ref{lemma51}, the two terms in the first line cancel each other, and the identity
$$\frac{d}{d\tau}(\tilde{k}_m^\tau\tilde{k}_{m'}^\tau)=\delta(m-m')\frac{d}{d\tau}\tilde{p}_m^\tau=\delta(m-m')(-\tilde{s}_m^\tau\mathscr{B}\tilde{p}_m^\tau-\tilde{p}_m^\tau\mathscr{B}\tilde{s}_m^\tau)$$
yields the relations
\[ \mathcal{A}_m^\tau=-\mathcal{C}_m^\tau\qquad\text{and}\qquad (\mathcal{D}_m^\tau)^*=-\mathcal{D}_m^\tau \:. \]
Thus
\beq\label{struc-deriv-ktil}
\frac{d}{d\tau}\tilde{k}_m^\tau=-\tilde{s}_m^\tau\mathscr{B}\tilde{k}_m^\tau-\tilde{k}_m^\tau\mathscr{B}\tilde{s}_m^\tau+\tilde{p}_m^\tau\mathcal{A}_m^\tau\tilde{p}_m^\tau-\tilde{k}_m^\tau\mathcal{A}_m^\tau\tilde{k}_m^\tau+\tilde{k}_m^\tau\mathcal{D}_m^\tau\tilde{p}_m^\tau-\tilde{p}_m^\tau\mathcal{D}_m^\tau\tilde{k}_m^\tau
\eeq
with
\beq\label{op-rels}
(\mathcal{A}_m^\tau)^* = \mathcal{A}_m^\tau
\qquad\text{and}\qquad (\mathcal{D}_m^\tau)^*=-\mathcal{D}_m^\tau \:.
\eeq

We now define the operators
\begin{align*}
A(\tau):=& \frac{1}{4 \pi^2} \lim_{\varepsilon \searrow 0}
\int_{C_\varepsilon} d\mu \int_{C_\varepsilon} d\nu \: \frac{\text{PP}}{\nu - \mu}\:
s^\tau_\mu \B s^\tau_\nu \\
&+\frac{1}{2}\int_0^\infty \Big(\tilde{p}_m^\tau\mathcal{D}_m^\tau\tilde{p}_m^\tau-\tilde{k}_m^\tau\mathcal{D}_m^\tau\tilde{k}_m^\tau+\tilde{k}_m^\tau\mathcal{A}_m^\tau\tilde{p}_m^\tau-\tilde{p}_m^\tau\mathcal{A}_m^\tau\tilde{k}_m^\tau´\Big) dm \:.
\end{align*}
From~\eqref{op-rels} it is obvious that~$A(\tau)$ is an anti-symmetric operator.
Furthermore, applying Lemma~\ref{lemma52}, we find that
\begin{align*}
\int_{C_\varepsilon} d\mu &\int_{C_\varepsilon} d\nu \: \frac{\text{PP}}{\nu - \mu}\:
s^\tau_\mu \B s^\tau_\nu\: \tilde{p}_m^\tau
= \int_{C_\varepsilon} d\mu \int_{C_\varepsilon} d\nu \: \frac{\text{PP}}{\nu - \mu}\:
s^\tau_\mu \B \: \frac{1}{m-\nu} \: \tilde{p}_m^\tau \\
&= -2 \pi i \int_{C_\varepsilon} d\mu\: \frac{\text{PP}}{m - \mu}\:
s^\tau_\mu \B \tilde{p}_m^\tau = -4 \pi^2 s^\tau_m \B \tilde{p}_m^\tau \:.
\end{align*}
Using this identity together with the analogous identity with~$\tilde{p}_m^\tau$
replaced by~$\tilde{k}_m^\tau$, one readily verifies that~\eqref{struc-deriv-ptil}
and~\eqref{struc-deriv-ktil} can be written as
$$\frac{d}{d\tau}\tilde{p}_m^\tau=A(\tau)\tilde{p}_m^\tau+\tilde{p}_m^\tau A(\tau)^*\qquad\text{and}\qquad\frac{d}{d\tau}\tilde{k}_m^\tau=A(\tau)\tilde{k}_m^\tau+\tilde{k}_m^\tau A(\tau)^* \:, $$
concluding the proof.
\QED

As mentioned at the beginning of this section, the result of Theorem \ref{pfeq} can be used to prove the existence of a one-parameter family of unitary transformations describing the interaction.
\begin{Thm}
There exists a one-parameter family of operators $U(\tau)$, which are unitary with respect to the
inner product~\eqref{iprod}, such that for any mass $m\geq0$ and for all $\tau\geq0$
\beq \label{eq:unit-thm}
	\tilde{p}_m^\tau=U(\tau)p_mU(\tau)^*\;\qquad\text{and}\qquad
	\tilde{k}_m^\tau=U(\tau)k_mU(\tau)^* \:.
\eeq
\end{Thm}
\Proof
If there exists such a family $U(\tau)$, then taking the $\tau$-derivative of \eqref{eq:unit-thm} yields 
\begin{align*}	\frac{d}{d\tau}\tilde{p}_m^\tau&=\frac{d}{d\tau}\Big(U(\tau)p_mU(\tau)^{-1}\Big)=\dot{U}(\tau)p_m(\tau)U(\tau)^{-1}-U(\tau)p_mU(\tau)^{-2}\dot{U}(\tau)\\&=\Big[\dot{U}(\tau)U(\tau)^{-1},\tilde{p}_m^\tau\Big]
\end{align*}
and
$$
\frac{d}{d\tau}\tilde{k}_m^\tau=\Big[\dot{U}(\tau)U(\tau)^{-1},\tilde{k}_m^\tau\Big].
$$
On the other hand, from Theorem \ref{pfeq} we know that the perturbation flow equations do have the commutator structure \eqref{pk-deriv-commutator}. Thus integrating these equations and taking ordered exponentials yields a family of unitary operators
\beq \label{formula-Utau}
U(\tau)=\Texp\Big(-\int_\tau^0 A(\tau)\,d\tau\Big)
\eeq
fulfilling the identities \eqref{eq:unit-thm}.
\QED

We remark that the operators $\mathcal{A}_m^\tau$ and $\mathcal{D}_m^\tau$ are  uniquely defined by perturbation series, for a perturbation being of the order~$\mathscr{O}(\B^1)$. The same is then true for the operators $A(\tau)$. This implies that the operator~$U(\tau)$ in~\eqref{formula-Utau} has a well-defined perturbation expansion.
Going through the combinatorial details, it would be straightforward to 
derive explicit formulas for this perturbation expansion to any order. Here we are content with the
existence statement which immediately implies the following interesting corollary.
\begin{Corollary}
There exists a unitary operator $U$, defined by a unique perturbation series, which for any mass $m > 0$ transforms the free fermionic projector \eqref{Psea} into the interacting fermionic projector \eqref{P-pq}:
\beq
P_m = U P_m^{\sea} U^* \:.
\eeq
\end{Corollary}

\section{The Light-Cone Expansion of the Fermionic Projector} \label{sec6}
The light-cone expansion is a useful technique for analyzing the fermionic projector near the light cone.
An introduction to this technique can be found in~\cite{PFP}, for details see~\cite{firstorder, light}.
In this section we will show that our rescaling procedure has no effect on the singularities of the
fermionic projector on the light cone, but it does change the regular
so-called high-energy contribution.

We briefly recall the definition of the light-cone expansion.
\newtheorem{def-lce}{Definition}[section]
\begin{def-lce}
A distribution $D(x,y)$ on $\mathbb{R}^4\times\mathbb{R}^4$ is of the
order $\mathcal{O}((y-x)^{2p})$ on the light cone, $p\in\mathbb{Z}$, if the product
$$(y-x)^{-2p}D(x,y)$$
is a regular distribution, i.e.\ a locally integrable function. It has the light cone expansion
\begin{align}
	D(x,y)=\sum_{j=g}^{\infty}D^{[j]}(x,y)
\end{align}
with $g\in\mathbb{Z}$, if the distributions $D^{[j]}(x,y)$ are of the order $\mathcal{O}((y-x)^{2j})$
on the light cone and if $D$ is approximated by the partial sums in the sense that 
\begin{align}
	D(x,y)-\sum_{j=g}^{p}D^{[j]}(x,y)
\end{align}
is of order $\mathcal{O}((y-x)^{2p+2})$ on the light cone for all $p\geq g$.
\end{def-lce} \noindent
We note that the lowest summand $D^{[g]}(x,y)$ gives the leading order on the light cone. If~$D$ is singular on the light cone, $g$ is negative.

By explicitly expanding the series~\eqref{series-scaustilde} term by term and
carrying out the sum over the orders in perturbation theory,
in~\cite{light} the light-cone expansion of the causal Green's functions~$\tilde{s}^\vee$
and~$\tilde{s}^\wedge$ is performed.
Using~\eqref{def-ktil}, one readily obtains the light-cone expansion of~$\tilde{k}$.
The so-called ``residual argument'' \cite[Chapter 3]{light} allows us to carry over
the light-cone expansion to the operator~$\tilde{p}^{\,\res}$, \eqref{ptil-thm}.
In this way, in~\cite{light} the light-cone expansion is derived for the so-called {\em{residual fermionic
projector}}~$P^{\,\res}$ defined by
\begin{align}
	P^{\,\res}:=\frac{1}{2}(\tilde{p}^{\,\res}-\tilde{k})\:.
	\label{resid-proj}
\end{align}
The residual fermionic projector is a solution of the Dirac equation \eqref{tdep}, but
its perturbation expansion has a different combinatorics than that of the fermionic projector $P$
as given by~\eqref{P-pq}. The difference of these two operators is defined to be
the high-energy contribution $P^{\he}(x,y)$,
\beq \label{Phedef}
P^{\,\he} :=P -P^{\,\res} \:.
\eeq
The following proposition extends the result of~\cite[Theorem~3.4]{light} to the rescaled
setting.
\begin{Prp}
Under the assumptions of Lemma \ref{regular-lemma}, the integral kernel~$P^{\he}(x,y)$ of the
high-energy contribution is, to every order in perturbation theory, a smooth function in~$x$ and~$y$.
\end{Prp}
\Proof
Using the formulas \eqref{P-pq} and \eqref{resid-proj} for $P$ and $P^{\,\res}$, one sees that 
$$P^{\,\he}=\frac{1}{2}(\tilde{k}-\tilde{k}^{\,\res}) \:.$$
As will be shown below, the following {\em{replacement rule}} holds:
\begin{quote}
If in the perturbation series  \eqref{qtil-thm} for $\tilde{k}^{\,\res}$
one replaces all factors $p$ by factors $k$, one obtains precisely the perturbation series~\eqref{k-pert}
for~$\tilde{k}$.
\end{quote}
Furthermore, we know from Lemma~\ref{lemma42} that the total number of
factors~$p$ in every summand of the perturbation expansion for~$\tilde{k}^{\,\res}$ is
even. Hence to every order in perturbation theory, $\tilde{k}$ can be obtained from~$\tilde{k}^{\,\res}$
by iteratively replacing pairs of factors~$p$ by factors~$k$.
Thus, to every order in perturbation theory, the operator~$P^{\,\he}$ can be written as a finite sum of
operator products of the form
\begin{align}	C^j_{n}\mathscr{B}\cdots C_{\sigma+1}\mathscr{B}\Big(k\mathscr{B}C_{\sigma-1}\cdots C_{\tau+1}\mathscr{B}k-p\mathscr{B}C_{\sigma-1}\cdots C_{\tau+1}\mathscr{B}p\Big)\mathscr{B}C_{\tau-1}\cdots \mathscr{B}C_0\:.
\label{elementaries}
\end{align}
It is shown in the proof of \cite[Theorem~3.4]{light} that terms of the form \eqref{elementaries} are smooth and bounded in position space, provided that the assumptions of Lemma \ref{regular-lemma} hold. 

It remains to prove the above replacement rule. For convenience, the
replacement of all factors~$p$ by factors~$k$ is denoted by the symbol~$\leadsto$.
Furthermore, we introduce the abbreviation
\begin{align} {\mathscr{S}}(l) :=
\Big[\sum_{\beta}(-i\pi)^{2\beta}k(bk)^{2\beta}(p+\pi^2pbpbp)\Big]^{2l}\sum_{\alpha}(-i\pi)^{2\alpha}k(bk)^{2\alpha}\:.
\label{opprods-in-sum}
\end{align}
In the proof of Lemma~\ref{lemma42}, we saw after~\eqref{trunc-sum} that for any~$r \in \N$,
\beq \label{trunc-sum2}
\tilde{k}^{\,\res} = \sum_{l=0}^r\, f_{l,r}\:b^<\:{\mathscr{S}}(l)\: b^>  + \O(b^{2r+2})\:.
\eeq
Replacing in~\eqref{opprods-in-sum} all factors~$p$ by~$k$, we obtain
\begin{align*}
{\mathscr{S}}(l) \leadsto & \Big[\sum_{\beta=0}^r
(-i\pi)^{2\beta}(kb)^{2\beta}(1+\pi^2kbkb)\Big]^{2l}
\sum_{\alpha=0}^r(-i\pi)^{2\alpha}k(bk)^{2\alpha} + \O(b^{2r+2}) \\
&=\sum_{\alpha=0}^r(-i\pi)^{2\alpha}k(bk)^{2\alpha} + \O(b^{2r+2}) \:,
\end{align*}
where in the last step we used that the sum in the square bracket is telescopic.
Substituting this formula into~\eqref{trunc-sum2}, we obtain
\[ \tilde{k}^{\,\res} \leadsto \sum_{l=0}^r\, f_{l,r}\:\sum_{\alpha=0}^r(-i\pi)^{2\alpha}
b^<\: k(bk)^{2\alpha}\: b^>  + \O(b^{2r+2})
= \sum_{l=0}^r\, f_{l,r}\:\tilde{k}  + \O(b^{2r+2}) \:, \]
where we used~\eqref{k-pert}. The result now follows from the identity
$$ \sum_{l=0}^rf_{l,r}=1\:, $$
which is verified by an elementary combinatorial argument.
\QED

According to~\eqref{Phedef} and the above Proposition, the singular behavior of the fermionic
projector $P$ is completely described by the residual fermionic projector $P^{\,\res}$.
In~\cite{light} the residual fermionic projector is decomposed into two parts,
\[	P^{\,\res}(x,y)=P^{\,\causal}(x,y)+P^{\,\lec}(x,y) \,. \]
Thus, the light-cone expansion of the fermionic projector becomes
\[	P(x,y)=P^{\,\causal}(x,y)+P^{\,\lec}(x,y)+P^{\,\he}(x,y) \,. \]
Here the contribution $P^{\,\causal}$ encodes the singular behavior of $P$ on the light cone,
whereas~$P^{\,\he}(x,y)$ and~$P^{\,\lec}(x,y)$ are, to every order in perturbation theory,
smooth functions in~$x$ and~$y$.

We finally explain how the ``causality'' of the causal perturbation expansion is to be understood.
Causality entered our construction when defining~$\tilde{k}$ in terms of the causal
Green's functions~\eqref{def-ktil}. This is reflected in the kernel~$P^{\,\causal}(x,y)$,
which is local and causal in the sense that it is given
explicitly in terms of integrals of~$\B$ and its partial derivatives
along the line segment~$\overline{xy}$.
However, taking the absolute value of an operator~\eqref{ptil-formally}
as well as the rescaling procedure~\eqref{eq:defP} do not preserve causality,
but these operations are necessary in order to distinguish a properly normalized
subspace of the Dirac solution space. The non-causality of our construction becomes apparent
in the high- and low-energy contributions, which involve integrals of~$\B$
over the whole space-time.
Thus our construction reveals that, although the Dirac equation is causal, the interacting Dirac sea
is a global object in space-time, violating locality and causality.
For a discussion of this remarkable fact we refer the reader to the book \cite{PFP}.

\begin{appendix}

\section{The Leading Orders of the Perturbation Expansions} \label{appA}
The operator $\tilde{k}$ in \eqref{k-pert} is given by
\begin{align*}	
\tilde{k} =&k-s\mathscr{B}k-k\mathscr{B}s+k\mathscr{B}s\mathscr{B}s+s\mathscr{B}k\mathscr{B}s+s\mathscr{B}s\mathscr{B}k-\pi^2k\mathscr{B}k\mathscr{B}k\\
&-k\mathscr{B}s\mathscr{B}s\mathscr{B}s-s\mathscr{B}k\mathscr{B}s\mathscr{B}s-s\mathscr{B}s\mathscr{B}k\mathscr{B}s-s\mathscr{B}s\mathscr{B}s\mathscr{B}k\\
&+\pi^2\Big(s\mathscr{B}k\mathscr{B}k\mathscr{B}k+k\mathscr{B}s\mathscr{B}k\mathscr{B}k+k\mathscr{B}k\mathscr{B}s\mathscr{B}k+k\mathscr{B}k\mathscr{B}k\mathscr{B}s\Big)+\O(\mathscr{B}^4).
\end{align*}

The operator $X$ in \eqref{X-2ndversion} has the expansion
\begin{align*}
X=&\pi^2\Big(p\mathscr{B}p\mathscr{B}p-p\mathscr{B}k\mathscr{B}k-k\mathscr{B}k\mathscr{B}p+k\mathscr{B}p\mathscr{B}k\\
&+p\mathscr{B}s\mathscr{B}k\mathscr{B}k-k\mathscr{B}p\mathscr{B}s\mathscr{B}k-p\mathscr{B}p\mathscr{B}s\mathscr{B}p+p\mathscr{B}k\mathscr{B}s\mathscr{B}k\\
&+k\mathscr{B}k\mathscr{B}s\mathscr{B}p-k\mathscr{B}s\mathscr{B}p\mathscr{B}k+k\mathscr{B}s\mathscr{B}k\mathscr{B}p-p\mathscr{B}s\mathscr{B}p\mathscr{B}p\Big)+\O(\mathscr{B}^4),
\end{align*}
which yields for the operator $\tilde{p}$ in \eqref{ptil-pert}
\begin{align*}
\tilde{p}=&p-s\mathscr{B}p-p\mathscr{B}s+p\mathscr{B}s\mathscr{B}s+s\mathscr{B}p\mathscr{B}s+s\mathscr{B}s\mathscr{B}p\\
&+\frac{\pi^2}{2}\Big(k\mathscr{B}p\mathscr{B}k-k\mathscr{B}k\mathscr{B}p-p\mathscr{B}k\mathscr{B}k-p\mathscr{B}p\mathscr{B}p\Big)\\
&-p\mathscr{B}s\mathscr{B}s\mathscr{B}s-s\mathscr{B}p\mathscr{B}s\mathscr{B}s-s\mathscr{B}s\mathscr{B}p\mathscr{B}s-s\mathscr{B}s\mathscr{B}s\mathscr{B}p\\
&+\frac{\pi^2}{2}\Big(s\mathscr{B}p\mathscr{B}p\mathscr{B}p+p\mathscr{B}s\mathscr{B}p\mathscr{B}p+p\mathscr{B}p\mathscr{B}s\mathscr{B}p+p\mathscr{B}p\mathscr{B}p\mathscr{B}s\\
&+k\mathscr{B}s\mathscr{B}k\mathscr{B}p+k\mathscr{B}k\mathscr{B}s\mathscr{B}p+p\mathscr{B}k\mathscr{B}s\mathscr{B}k+p\mathscr{B}s\mathscr{B}k\mathscr{B}k-k\mathscr{B}s\mathscr{B}p\mathscr{B}k\\
&-k\mathscr{B}p\mathscr{B}k\mathscr{B}s-k\mathscr{B}p\mathscr{B}s\mathscr{B}k+p\mathscr{B}k\mathscr{B}k\mathscr{B}s+s\mathscr{B}p\mathscr{B}k\mathscr{B}k+s\mathscr{B}k\mathscr{B}k\mathscr{B}p\\
&+k\mathscr{B}k\mathscr{B}p\mathscr{B}s-s\mathscr{B}k\mathscr{B}p\mathscr{B}k\Big)+\O(\mathscr{B}^4)\:,
\end{align*}
and for the rescaling operator $Y$ in \eqref{def-x}
\begin{align*}
Y=&p-\pi^2\Big(p\mathscr{B}p\mathscr{B}p-p\mathscr{B}k\mathscr{B}k-k\mathscr{B}k\mathscr{B}p+k\mathscr{B}p\mathscr{B}k\\
&+p\mathscr{B}s\mathscr{B}k\mathscr{B}k-k\mathscr{B}p\mathscr{B}s\mathscr{B}k-p\mathscr{B}p\mathscr{B}s\mathscr{B}p+p\mathscr{B}k\mathscr{B}s\mathscr{B}k\\
&+k\mathscr{B}k\mathscr{B}s\mathscr{B}p-k\mathscr{B}s\mathscr{B}p\mathscr{B}k+k\mathscr{B}s\mathscr{B}k\mathscr{B}p-p\mathscr{B}s\mathscr{B}p\mathscr{B}p\Big)+\O(\mathscr{B}^4).
\end{align*}

The operator $\tilde{p}^{\,\res}$ in \eqref{p-res} has the expansion
\begin{align*}	
\tilde{p}^{\,\res} =&p-s\mathscr{B}p-p\mathscr{B}s+p\mathscr{B}s\mathscr{B}s+s\mathscr{B}p\mathscr{B}s+s\mathscr{B}s\mathscr{B}p-\pi^2p\mathscr{B}p\mathscr{B}p\\
&-p\mathscr{B}s\mathscr{B}s\mathscr{B}s-s\mathscr{B}p\mathscr{B}s\mathscr{B}s-s\mathscr{B}s\mathscr{B}p\mathscr{B}s-s\mathscr{B}s\mathscr{B}s\mathscr{B}p\\
&+\pi^2\Big(s\mathscr{B}p\mathscr{B}p\mathscr{B}p+p\mathscr{B}s\mathscr{B}p\mathscr{B}p+p\mathscr{B}p\mathscr{B}s\mathscr{B}p+p\mathscr{B}p\mathscr{B}p\mathscr{B}s\Big)+\O(\mathscr{B}^4) ,
\end{align*}
whereas the operator $\tilde{k}^{\,\res}$ in \eqref{eq:qtilde} is given by
\begin{align*}	
\tilde{k}^{\,\res} =&k-s\mathscr{B}k-k\mathscr{B}s+k\mathscr{B}s\mathscr{B}s+s\mathscr{B}k\mathscr{B}s+s\mathscr{B}s\mathscr{B}k\\
&+\frac{\pi^2}{2}\Big(-k\mathscr{B}p\mathscr{B}p+p\mathscr{B}k\mathscr{B}p-p\mathscr{B}p\mathscr{B}k-k\mathscr{B}k\mathscr{B}k\Big)\\
&-k\mathscr{B}s\mathscr{B}s\mathscr{B}s-s\mathscr{B}k\mathscr{B}s\mathscr{B}s-s\mathscr{B}s\mathscr{B}k\mathscr{B}s-s\mathscr{B}s\mathscr{B}s\mathscr{B}k\\
&+\frac{\pi^2}{2}\Big(s\mathscr{B}k\mathscr{B}k\mathscr{B}k+k\mathscr{B}s\mathscr{B}k\mathscr{B}k+k\mathscr{B}k\mathscr{B}s\mathscr{B}k+k\mathscr{B}k\mathscr{B}k\mathscr{B}s\\
&+k\mathscr{B}s\mathscr{B}p\mathscr{B}p-s\mathscr{B}p\mathscr{B}k\mathscr{B}p+s\mathscr{B}p\mathscr{B}p\mathscr{B}k+p\mathscr{B}s\mathscr{B}p\mathscr{B}k-p\mathscr{B}k\mathscr{B}p\mathscr{B}s\\
&+k\mathscr{B}p\mathscr{B}s\mathscr{B}p+k\mathscr{B}p\mathscr{B}p\mathscr{B}s-p\mathscr{B}k\mathscr{B}s\mathscr{B}p+p\mathscr{B}p\mathscr{B}k\mathscr{B}s+s\mathscr{B}k\mathscr{B}p\mathscr{B}p\\
&-p\mathscr{B}s\mathscr{B}k\mathscr{B}p+p\mathscr{B}p\mathscr{B}s\mathscr{B}k\Big)+\O(\mathscr{B}^4).
\end{align*}

Thus for the operator $\tilde{t}$ in \eqref{eq:tt} we obtain
\begin{align*}
	\tilde{t}=\frac{1}{2}\Bigg[&
	p-k-s\mathscr{B}p-p\mathscr{B}s+s\mathscr{B}k+k\mathscr{B}s\\	&+p\mathscr{B}s\mathscr{B}s+s\mathscr{B}p\mathscr{B}s+s\mathscr{B}s\mathscr{B}p-k\mathscr{B}s\mathscr{B}s-s\mathscr{B}k\mathscr{B}s-s\mathscr{B}s\mathscr{B}k\\	&+\pi^2k\mathscr{B}k\mathscr{B}k+\frac{\pi^2}{2}\Big(k\mathscr{B}p\mathscr{B}k-k\mathscr{B}k\mathscr{B}p-p\mathscr{B}k\mathscr{B}k-p\mathscr{B}p\mathscr{B}p\Big)\\	&-p\mathscr{B}s\mathscr{B}s\mathscr{B}s-s\mathscr{B}p\mathscr{B}s\mathscr{B}s-s\mathscr{B}s\mathscr{B}p\mathscr{B}s-s\mathscr{B}s\mathscr{B}s\mathscr{B}p\\
&+k\mathscr{B}s\mathscr{B}s\mathscr{B}s+s\mathscr{B}k\mathscr{B}s\mathscr{B}s+s\mathscr{B}s\mathscr{B}k\mathscr{B}s+s\mathscr{B}s\mathscr{B}s\mathscr{B}k\\
&+\frac{\pi^2}{2}\Big(s\mathscr{B}p\mathscr{B}p\mathscr{B}p+p\mathscr{B}s\mathscr{B}p\mathscr{B}p+p\mathscr{B}p\mathscr{B}s\mathscr{B}p+p\mathscr{B}p\mathscr{B}p\mathscr{B}s\Big)\\
&-\pi^2\Big(s\mathscr{B}k\mathscr{B}k\mathscr{B}k+k\mathscr{B}s\mathscr{B}k\mathscr{B}k+k\mathscr{B}k\mathscr{B}s\mathscr{B}k+k\mathscr{B}k\mathscr{B}k\mathscr{B}s\Big)\\
&+\frac{\pi^2}{2}\Big(k\mathscr{B}s\mathscr{B}k\mathscr{B}p+k\mathscr{B}k\mathscr{B}s\mathscr{B}p+p\mathscr{B}k\mathscr{B}s\mathscr{B}k+p\mathscr{B}s\mathscr{B}k\mathscr{B}k\\
&-k\mathscr{B}s\mathscr{B}p\mathscr{B}k-k\mathscr{B}p\mathscr{B}k\mathscr{B}s-k\mathscr{B}p\mathscr{B}s\mathscr{B}k+p\mathscr{B}k\mathscr{B}k\mathscr{B}s\\
&+s\mathscr{B}p\mathscr{B}k\mathscr{B}k+s\mathscr{B}k\mathscr{B}k\mathscr{B}p+k\mathscr{B}k\mathscr{B}p\mathscr{B}s-s\mathscr{B}k\mathscr{B}p\mathscr{B}k\Big)\Bigg]+\O(\mathscr{B}^4),
\end{align*}
and the rescaled fermionic projector \eqref{P-pq} is given by
\begin{align*}	
P=\frac{1}{2}\Bigg[&p-k-s\mathscr{B}p-p\mathscr{B}s+s\mathscr{B}k+k\mathscr{B}s\\
&+p\mathscr{B}s\mathscr{B}s+s\mathscr{B}p\mathscr{B}s+s\mathscr{B}s\mathscr{B}p
-k\mathscr{B}s\mathscr{B}s-s\mathscr{B}k\mathscr{B}s-s\mathscr{B}s\mathscr{B}k\\
&-\pi^2p\mathscr{B}p\mathscr{B}p+\frac{\pi^2}{2}\Big(k\mathscr{B}k\mathscr{B}k+p\mathscr{B}p\mathscr{B}k-p\mathscr{B}k\mathscr{B}p+k\mathscr{B}p\mathscr{B}p\Big)\\
&-p\mathscr{B}s\mathscr{B}s\mathscr{B}s-s\mathscr{B}p\mathscr{B}s\mathscr{B}s-s\mathscr{B}s\mathscr{B}p\mathscr{B}s-s\mathscr{B}s\mathscr{B}s\mathscr{B}p\\
&+k\mathscr{B}s\mathscr{B}s\mathscr{B}s+s\mathscr{B}k\mathscr{B}s\mathscr{B}s+s\mathscr{B}s\mathscr{B}k\mathscr{B}s+s\mathscr{B}s\mathscr{B}s\mathscr{B}k\\
&+\pi^2\Big(s\mathscr{B}p\mathscr{B}p\mathscr{B}p+p\mathscr{B}s\mathscr{B}p\mathscr{B}p+p\mathscr{B}p\mathscr{B}s\mathscr{B}p+p\mathscr{B}p\mathscr{B}p\mathscr{B}s\Big)\\
&-\frac{\pi^2}{2}\Big(s\mathscr{B}k\mathscr{B}k\mathscr{B}k+k\mathscr{B}s\mathscr{B}k\mathscr{B}k+k\mathscr{B}k\mathscr{B}s\mathscr{B}k+k\mathscr{B}k\mathscr{B}k\mathscr{B}s\Big)\\
&+\frac{\pi^2}{2}\Big(p\mathscr{B}k\mathscr{B}p\mathscr{B}s-s\mathscr{B}k\mathscr{B}p\mathscr{B}p+p\mathscr{B}k\mathscr{B}s\mathscr{B}p-p\mathscr{B}p\mathscr{B}k\mathscr{B}s\\
&-p\mathscr{B}s\mathscr{B}p\mathscr{B}k+p\mathscr{B}s\mathscr{B}k\mathscr{B}p-k\mathscr{B}p\mathscr{B}s\mathscr{B}p-k\mathscr{B}s\mathscr{B}p\mathscr{B}p\\
&-p\mathscr{B}p\mathscr{B}s\mathscr{B}k+s\mathscr{B}p\mathscr{B}k\mathscr{B}p-k\mathscr{B}p\mathscr{B}p\mathscr{B}s-s\mathscr{B}p\mathscr{B}p\mathscr{B}k\Big)\Bigg]+\O(\mathscr{B}^4).
\end{align*}

\end{appendix}

\Thanks{{\em{Acknowledgments:}}  A.G.\ would like to thank the Erwin Schr\"odinger Institute,
Vienna, for its hospitality while he was working on the manuscript.}


\providecommand{\bysame}{\leavevmode\hbox to3em{\hrulefill}\thinspace}
\providecommand{\MR}{\relax\ifhmode\unskip\space\fi MR }
\providecommand{\MRhref}[2]{%
  \href{http://www.ams.org/mathscinet-getitem?mr=#1}{#2}
}
\providecommand{\href}[2]{#2}

\end{document}